\documentclass[11pt]{article}
\usepackage{amssymb}
\usepackage{graphicx}
\usepackage{amsmath}
\usepackage{makeidx}
\usepackage{indentfirst}

\newcounter{resultnum}[section]
\newcounter{conclusionnum}[section]
\newcounter{conditionnum}[section]
\newcounter{conjecturenum}[section]
\newcounter{examplenum}[section]
\newcounter{exercisenum}[section]
\newtheorem{lemma}{Lemma}[section]

\newcounter{lemmanum}[section]
\newcounter{notationnum}[section]
\newtheorem{theorem}{Theorem}[section]

\newcounter{theoremnum}[section]
\newtheorem{definition}{Definition}[section]

\newcounter{definitionnum}[section]
\newtheorem{corollary}{Corollary}[section]

\newcounter{corollarynum}[section]
\newtheorem{remark}{Remark}[section]

\newcounter{remarknum}[section]
\newcounter{propositionnum}[section]
\newcounter{acknowledgementnum}[section]
\newcounter{algorithmnum}[section]
\newcounter{axiomnum}[section]
\newcounter{casenum}[section]
\newcounter{claimnum}[section]
\newcounter{summarynum}[section]
\newcounter{problemnum}[section]
\newenvironment{proof}[1][]{\textbf{Proof.} }{}

\begin{document}

\title{Hidden Symmetries for Ellipsoid--Solitonic Deformations of Kerr--Sen
Black Holes and Quantum Anomalies}
\date{December 15, 2012}
\author{ Sergiu I. Vacaru\thanks{
sergiu.vacaru@uaic.ro; http://www.scribd.com/people/view/1455460-sergiu} \\
{\small {\textsl{\ Science Department, University "Al. I. Cuza" Ia\c si},} }%
\\
{\small {\textsl{\ 54 Lascar Catargi street, 700107, Ia\c si, Romania}} } }
\maketitle

\begin{abstract}
We prove the existence of hidden symmetries in the general relativity theory
defined by exact solutions with generic off--diagonal metrics, nonholonomic
(non--integrable) constraints, and deformations of the frame and linear
connection structure. A special role in characterization of such spacetimes
is played by the corresponding nonholonomic generalizations of
Stackel--Killing and Killing--Yano tensors. There are constructed new
classes of black hole solutions and studied hidden symmetries for
ellipsoidal and/or solitonic deformations of "prime" Kerr--Sen black holes
into "target" off--diagonal metrics. In general, the classical conserved
quantities (integrable and not--integrable) do not transfer to the quantized
systems and produce quantum gravitational anomalies. We prove that such
anomalies can be eliminated via corresponding nonholonomic deformations of
fundamental geometric objects (connections and corresponding Riemannian and
Ricci tensors) and by frame transforms.

\vskip0.1cm



PACS2008:\ 04.90.+e, 04.40.Nr, 11.30.Ly, 11.30.Na
\end{abstract}



\section{Introduction}

A very important task in the study of classical and quantum gravitational
and matter field interactions is to determine the corresponding symmetries
of dynamical systems and to identify the constants of motion and related
conservation laws. In general, the evolution of a dynamical system is
described in the phase-space when nonholonomic (equivalently, anholonomic
and/or non--integrable) constraints are imposed on dynamical variables.
Various methods of nonholonomic geometry are considered in classical
Lagrange/Hamilton mechanics, in quantum field theory of gauge fields and,
for instance, in the Dirac approach to perturbative quantum gravity.

It is natural to search of conserved quantities corresponding to
nonholonomically deformed and/or hidden symmetries of the complete
phase--space, not just for the configuration one. For instance, a
fundamental hidden symmetry for rotating black hole spacetimes with
spherical horizon topology and Taub--NUT solutions is the Killing--Yano (KY)
one \cite{yano1}, see details in references in \cite{chen,page,visin1}.
There were considered various types of generalized and associated hidden
symmetries. A large class of symmetries is characterized by higher rank
symmetric Stackel--Killing (SK) tensors which generalize the Killing vectors.%
\footnote{%
Such higher order symmetries are called "hidden symmetries" and the
corresponding holonomic (un--constrained) values are quadratic, or (in
general) polynomial in momenta. In this work, we study more general
gravitational nonlinear systems with non--integrable constraints and
"nonholonomic hidden symmetries".} There were considered also antisymmetric
KY tensors, corresponding conformal extensions (CSK and CKY) etc, see
reviews of results and applications in \cite{sb,frolov,vh}.

Quite generally, this paper is connected with three recent directions in
modern string gravity, supergravity and applications of Ricci flow theory in
physics. The first one is with extensions of the KY symmetry in the presence
of skew--symmetric torsion \cite{yb,kub,ahmedov,houri}. The second one was
proposed in two papers on anholonomic frames, generalized Killing equations
and Dirac operators considered on anisotropic Taub NUT spinning spaces \cite%
{vp,vt}. The third direction is related to some geometric methods in the
geometry of nonholonomic Ricci flows and constructing wave type solutions
for Einstein and/or Finsler spaces \cite{vex2,vex3}.

In our recent works, there were proved two important results: 1) Using the
anholonomic deformation method \cite{vex2,vex3}, it is possible to decouple
(equivalently, separate) the Einstein equations and generate exact solutions
in very general forms. 2) Working with nonholonomic deformations of the
Levi--Civita connection to certain auxiliary connections (also uniquely
defined by the metric and/or corresponding almost symplectic structure), the
Einstein gravity and generalizations were quantized in certain equivalent
forms using deformation quantization, A--brane formalism and two connection
renormalization \cite{vdq3,vabr}. Such constructions are characterized by
special types of symmetries derived for the corresponding class of
admissible nonholonomic deformations of geometric and physical objects in
classical and quantum gravity. This motivates our study of associated hidden
symmetries determined by nonholonomic deformations of KY and SK tensors.

The nonholonomic geometries with induced (by metric structure) torsions have
some similarities, for instance, with possible extensions (black hole
spacetimes, various supergravity theories etc) of the KY symmetry in the
presence of skew--symmetric torsion \cite%
{yb,kub,ahmedov,wu,chong,strom,agricola}. They also admit generalized KY
tensors with corrections which became trivial if torsion fields vanish. In
order to study the existing common features and differences of spacetimes
with ''generic'' torsions and induced (nonholonomically) torsions (which are
equivalent to the field equations in GR), we provide explicit constructions
and analyze the basic properties of nonholonomic KY tensors and theirs CKY
extensions.

The results on exact off--diagonal solutions and quantum gravitational field
theories with nontrivial nonholonomic constraints give rise to two natural
questions: 1) wether the nonholonomic KY and/or SK symmetries are relied, or
not, to certain fundamental properties of vacuum and non--vacuum
gravitational, and gravitational--matter fields interactions and 2) if there
are some interesting and physically important spacetimes characterized by
such nonholonomic symmetries? One of the aims of this work is to show that
nonholonomic hidden symmetries exist naturally for any generic off--diagonal
spacetime and associated nonholonomic frame structures. We shall provide
explicit examples of exact solutions (for black ellipsoids, nonholonomic
deformations and gravitational and solitons) and construct/study the
corresponding nonholonomic generalization of KY and SK tensors.

The above mentioned types of nonholonomic hidden symmetries are fundamental
ones characterizing generic off--diagonal and non--integrable classical
gravitational interactions. Passing to quantized systems it is necessary to
investigate if such symmetries "survive", or not, and to perform a rigorous
study of the corresponding possible conserved quantities and separability of
the field/motion equations can be preserved. It is well known that in the
cases of hidden symmetries for the Levi--Civita connection there are
anomalies representing discrepancies between the conservation laws the the
classical level and the corresponding ones at the quantum level \cite%
{carter1,ivv}. One of main goals of this work is to prove that certain
classes of anomalies can be eliminated via corresponding nonholonomic
deformations and frame transforms. Such nonholonomic hidden symmetries exist
for adapted distinguished connections, they are not violated under certain
quantization schemes, being determined completely by the components of the
metric tensor. The constructions can be redefined equivalently in terms of
the "standard" Levi--Civita connection but the corresponding formulas are
quite sophisticate and less obvious then the nonholonomic ones if there are
considered general nonlinear classical and quantum gravitational systems.

The plan of the paper is as follows:\ In section \ref{s2} we outline the
anholonomic deformation method of constructing generic off--diagonal
solutions in gravity theories. We show how the Einstein equations can be
reformulated equivalently in nonholonomic variables which allows a ''magic
splitting'' into subsystems which can be integrated in very general forms.
In section \ref{s4}, the hidden symmetries are studied for explicit examples
of exact solutions described by generic off--diagonal metrics. There are
considered the nonholonomic symmetries for Kerr--Sen black ellipsoid/hole
configurations and their solitonic deformations.

In section \ref{s5} we prove that nonholonomic hidden symmetries for the
Einstein gravity and certain metric compatible generalizations survive under
transitions from classical to quantum spacetimes, which is very different
for the case of holonomic analogs with the Levi--Civita connection. We
conclude that the gravitational anomalies can be canceled via nonholonomic
deformations for very general nonlinear gravitational systems. Finally,
section \ref{s6} is devoted to conclusions. Some necessary coefficient
formulas are presented in Appendix.

\section{Nonholonomic Variables and Exact Solutions}

\label{s2} In general relativity theory (in brief, GR) the geometry
nonlinear and ''non--integrable'' interactions of particles and fields on a
(pseudo) Riemannian spacetime $V, $ $\dim V=4,$ (endowed with a metric
tensor $\mathbf{g}$ of signature $(-1,1,1,1)$), is encoded into certain
structures of moving frames (equivalently, tetrads/ vierbeins) $e_{\alpha
}=e_{\ \alpha}^{\underline{\alpha }}(u)\partial /\partial u^{\underline{%
\alpha }}$ and nonholonomic constraints $e_{\alpha }e_{\beta }-e_{\beta
}e_{\alpha }=W_{\alpha \beta }^{\gamma}e_{\gamma }$. The anholonomy
coefficients $W_{\alpha \beta }^{\gamma }=W_{\alpha \beta }^{\gamma }(u)$
vanish for holonomic, i. e. integrable, configurations. We shall follow such
conventions: the small Greek indices $\alpha ,\beta ,..$ can be abstract or
coordinate ones (we shall use also primed/underlined etc indices for
distinguishing different types of (non)holonomic frame decompositions).
Indices may be underlined, $\underline{\alpha }, \underline{\beta },...$ in
order to emphasize that they are coordinate ones, running values $\underline{%
\alpha },\underline{\beta }...=1,2,3,4$; we shall omit, for simplicity, any
priming/ underlying etc if that will not result in ambiguities. On
convenience, we shall consider also ''permutations'' of signatures with
metrics locally parametrized in the form $(1,-1,1,1),\ (1,1,-1,1),$ or $%
(1,1,1,-1).$ The local coordinates of a point $u\in V$ are labeled $%
u^{\beta},$ partial derivatives are $\partial _{\beta }:=\partial /\partial
u^{\beta}$ and functions are written as $f(u)\equiv f(u^{\alpha }).$

Under frame transforms, the coefficients of a metric $\mathbf{g=}g_{\alpha
\beta }e^{\alpha }$ $\otimes e^{\beta },$ for a dual $e^{\alpha },$ for $%
e^{\alpha }\rfloor e_{\beta }=\delta _{\beta }^{\alpha }$ (where the 'hook'
operator $\rfloor $ corresponds to the inner derivative and $\delta _{\beta
}^{\alpha }$ is the Kronecker symbol) are re--defined following the rule%
\begin{equation}
g_{\alpha \beta }=e_{\ \alpha }^{\underline{\alpha }}e_{\ \beta }^{%
\underline{\beta }}g_{\underline{\alpha }\underline{\beta }}.
\label{metrtransf}
\end{equation}%
On a spacetime manifold $V$, we can consider various types of \textit{%
nonholonomic structures} parametrized by some sets of coefficients $\{e_{\
\alpha }^{\underline{\alpha }}:\partial _{\underline{\alpha }}$ $\rightarrow
e_{\alpha }\}.$

\subsection{Manifolds with associated N--connections}

In a general context, we can use the definition of nonholonomic manifold $%
\mathcal{V}=(V,\emph{N}),$ where $V$ is a pseudo--Riemannian spacetime and $%
\mathcal{N}$ is a nonholonomic distribution.\footnote{%
Such a distribution can be stated by an arbitrary function (or a set of
functions) on $V$ prescribing a vierbein structure $e_{\ \alpha }^{%
\underline{\alpha }} $ following certain geometric principles. In modern
gravity, it is largely used the so--called ADM (Arnowit--Deser--Misner)
splitting, 3+1, see details in \cite{misner}. For our purposes, it is
convenient to work with an alternative non--integrable 2+2 splitting, which
allows us to decouple the Einstein equations and integrate them in "very"
general forms \cite{vex2,vex3}. Such a technique of generating exact
solutions can not be elaborated working only with 3+1 decompositions.} In
this paper, we shall consider distributions defining a $2+2$ splitting, i.e.
a non--integrable fibration, with associated nonlinear connection
(N--connection) structure $\emph{N}=\mathbf{N}$ satisfying the properties:

\begin{enumerate}
\item A N--connection is introduced as a Whitney sum $\mathbf{N}:T\mathbf{V}%
=h\mathbf{V}\oplus v\mathbf{V}$ defining a conventional horizontal (h) and
vertical (v) splitting.

\item N--adapted parametrizations of local coordinates and frames and
respective tensor indices on $\mathbf{V}$ are considered for $u^{\alpha
}=(x^{i},y^{a}),$ where h--indices take values $i,j,....=1,2$ and v--indices
take values $a,b,...=3,4,$ for $a=i+1$ and $b=j+2$ respectively contracted
with $i $ and $j$. This introduces a local fibred structure when the
coefficients of N--connection, $N_{i}^{a},$ for $\mathbf{N}%
=N_{i}^{a}(u)dx^{i}\otimes \partial /\partial y^{a}$, define N--adapted
/--elongated local bases (partial derivatives), $\mathbf{e}_{\nu }=(\mathbf{e%
}_{i},e_{a}),$ and cobases (differentials), $\mathbf{e}^{\mu }=(e^{i},%
\mathbf{e}^{a}),$ when
\begin{eqnarray}
&&\mathbf{e}_{i}=\frac{\partial }{\partial x^{i}}-\ N_{i}^{a}\frac{\partial
}{\partial y^{a}},\ e_{a}=\frac{\partial }{\partial y^{a}},  \label{nader} \\
\mbox{ and } &&e^{i}=dx^{i}, \mathbf{e}^{a}=dy^{a}+\ N_{i}^{a}dx^{i}.
\label{nadif}
\end{eqnarray}
We shall use boldface symbols for spaces enabled with N--connection
structure.

\item In general a metric structure
\begin{equation}
\ \mathbf{g}=\ \underline{g}_{\alpha \beta }\left( u\right) du^{\alpha
}\otimes du^{\beta }  \label{fmetr}
\end{equation}%
on $\mathbf{V}$ is generic off--diagonal, i. e. it can not be diagonalized
via coordinate transforms, and can be always parametrized (up to certain
classes of frame/coordinate transforms of type (\ref{metrtransf})) as%
\begin{equation}
\ \ \underline{g}_{\alpha \beta }\left( u\right) =\left[
\begin{array}{cc}
\ g_{ij}+\ h_{ab}N_{i}^{a}N_{j}^{b} & h_{ae}N_{j}^{e} \\
\ h_{be}N_{i}^{e} & \ h_{ab}%
\end{array}%
\right].  \label{fansatz}
\end{equation}
Equivalently, such a metric can be represented in N--adapted \ form (as a
distinguished metric/tensor, \ d--metric, d--tensor)
\begin{equation}
\ \mathbf{g}=\ g_{ij}(x,y)\ e^{i}\otimes e^{j}+\ g_{ab}(x,y)\mathbf{e}%
^{a}\otimes \mathbf{e}^{b}.  \label{dm}
\end{equation}
\end{enumerate}

A parametrization of type (\ref{fansatz}) should be not confused with
similar ones in the Kaluza--Klein theory when $y^{a}$ are considered as
extra dimension coordinates on which, for instance, a ''cylindrical''
compactification is performed and $N_{j}^{e}(x,y)\sim A_{aj}^{e}(x^{k})y^{a}$
are certain (non) Abelian gauge fields. In our approach, $N_{j}^{e}$ are
with arbitrary nonlinear dependence on $y^{\underline{a}}\rightarrow y^{a}$
and state naturally certain frame coefficients $e_{\ \alpha }^{\underline{%
\alpha }}$ with 2+2 splitting and linear elongations on $N^a_i$ in $\mathbf{e%
}_{\nu }$ (\ref{nader}) and $\mathbf{e}^{\mu }$ (\ref{nadif}). This subclass
of anholonomic frames is subjected to conditions
\begin{equation}
\lbrack \mathbf{e}_{\alpha },\mathbf{e}_{\beta }]=\mathbf{e}_{\alpha }%
\mathbf{e}_{\beta }-\mathbf{e}_{\beta }\mathbf{e}_{\alpha }=W_{\alpha \beta
}^{\gamma }\mathbf{e}_{\gamma },  \label{anhrel1}
\end{equation}
when \ the (antisymmetric) nontrivial anholonomy coefficients $%
W_{ia}^{b}=\partial _{a}N_{i}^{b}$ and $W_{ji}^{a}=\Omega _{ij}^{a}$ are
determined respectively by partial v--derivatives of $N_{i}^{b}$ and the
coefficients of curvature of N--connection $\Omega _{ij}^{a}=\mathbf{e}%
_{j}\left( N_{i}^{a}\right) -\mathbf{e}_{i}\left( N_{j}^{a}\right) .$ Having
prescribed a N--connection 2+2 splitting, we can say that our (pseudo)
Riemannian spacetime is modelled as N--anholonomic manifold (a similar
situation exists for the ADM formalism when the role of N--coefficients is
played by the so--called ''shift'' and ''lapse'' functions).

\subsection{Nonholonomic deformations of linear connections}

From the class of linear connections which can be defined on a manifold $%
\mathbf{V}$, we can select a subclass which is adapted to the N--connection
structure.

\begin{definition}
A distinguished connection, d--connection, $\mathbf{D=}(hD,vD)$ on a
N--anholonomic $\mathbf{V}$ is defined as a linear connection preserving
under parallelism the N--connection structure.
\end{definition}

The N--adapted components $\mathbf{\Gamma }_{\ \beta \gamma }^{\alpha }$ of
a d--connection $\mathbf{D}$ are computed following equations $\mathbf{D}%
_{\alpha }\mathbf{e}_{\beta }=\mathbf{\Gamma }_{\ \alpha \beta }^{\gamma }%
\mathbf{e}_{\gamma }$ and parametrized in the form $\ \mathbf{\Gamma }_{\
\alpha \beta }^{\gamma }=\left(
L_{jk}^{i},L_{bk}^{a},C_{jc}^{i},C_{bc}^{a}\right) ,$ where $\mathbf{D}%
_{\alpha }=(D_{i},D_{a}),$ with $h\mathbf{D}=(L_{jk}^{i},L_{bk}^{a})$ and $v%
\mathbf{D}=(C_{jc}^{i},$ $C_{bc}^{a})$ determining covariant, respectively,
h-- and v--derivatives. We can associate a differential 1--form $\mathbf{%
\Gamma }_{\ \beta }^{\alpha }=\mathbf{\Gamma }_{\ \beta \gamma }^{\alpha }%
\mathbf{e}^{\gamma }$ and perform a N--adapted differential form calculus.

The torsion $\mathcal{T}^{\alpha }=\{\mathbf{T}_{\ \beta \gamma }^{\alpha
}\} $ and curvature $\mathcal{R}_{~\beta }^{\alpha }=\{\mathbf{\mathbf{R}}%
_{\ \ \beta \gamma \delta }^{\alpha }\}$ of $\ \mathbf{D}$ are computed
respectively
\begin{eqnarray}
\ ^{\mathbf{D}}\mathcal{T}^{\alpha }&:=&\mathbf{De}^{\alpha }=d\mathbf{e}%
^{\alpha }+\mathbf{\Gamma }_{\ \beta }^{\alpha }\wedge \mathbf{e}^{\beta }
\label{tors} \\
\mbox{ and } \ ^{\mathbf{D}}\mathcal{R}_{~\beta }^{\alpha } &:=& \mathbf{%
D\Gamma }_{\ \beta }^{\alpha }=d\mathbf{\Gamma }_{\ \beta }^{\alpha }-%
\mathbf{\Gamma }_{\ \beta }^{\gamma }\wedge \mathbf{\Gamma }_{\ \gamma
}^{\alpha }=\mathbf{R}_{\ \beta \gamma \delta }^{\alpha }\mathbf{e}^{\gamma
}\wedge \mathbf{e}^{\delta } ,  \label{curv}
\end{eqnarray}%
see Refs. \cite{vex3} for \ explicit calculi of coefficients (we provide
some necessary component formulas in Appendix).\footnote{%
We shall use certain left ''up'' or ''low'' labels in order to emphasize
that certain geometric objects are determined by another fundamental
geometric object, \ for instance, that the torsion$\ ^{\mathbf{D}}\mathcal{T}%
^{\alpha }$ is determined by d--connection $\mathbf{D}$. We shall omit such
labels if that will not result in ambiguities.}

For any metric $\ \mathbf{g}$ (\ref{fmetr}) (equivalently, d--metric (\ref%
{dm})), we can introduce the Levi--Civita (LC) connection $\nabla =\{\
_{\shortmid }\Gamma _{\ \beta \gamma }^{\alpha }\}$ as the unique one
satisfying the metric compatibility condition, $\nabla \mathbf{g}=0,$ and
zero torsion \ condition, $\ ^{\nabla }\mathcal{T}^{\alpha }:=0.$ Such a
linear connection is not a d--connection because it does not preserve under
frame/coordinate/parallel transforms a prescribed N--connection splitting.
Nevertheless, it is possible to construct decompositions of type $\ ^{%
\mathbf{g}}\nabla =\ ^{\mathbf{g}}\mathbf{D}+\ ^{\mathbf{g}}\mathbf{Z,}$ for
$\ ^{\mathbf{g}}\mathbf{Dg=0,}$ when both the linear connection $\ ^{\mathbf{%
g}}\nabla $ and metric compatible d--connection $\ ^{\mathbf{g}}\mathbf{D}$
and the distortion tensor $\ ^{\mathbf{g}}\mathbf{Z}$ are defined in a
unique form, following well defined geometric principles, by a metric tensor
$\mathbf{g}$ when a N--connection structure $\mathbf{N}$ is presecribed on a
N--anholonomic spacetime $\mathbf{V.}$

\begin{theorem}
--\textbf{Definition.\ } \ There is a canonical d--connection $\widehat{%
\mathbf{D}}$ completely and uniquely defined by a (pseudo) Riemannian metric
$\ \mathbf{g}$ (\ref{fmetr}) for a chosen $\mathbf{N=\{}N_{i}^{a}\}$ if and
only if $\ \widehat{\mathbf{D}}\mathbf{g}=0$ and the horizontal and vertical
torsions are zero, i.e. $h\widehat{\mathcal{T}}=\{\mathbf{T}_{\ jk}^{i}\}=0$
and $v\widehat{\mathcal{T}}=\{\mathbf{T}_{\ bc}^{a}\}=0.$
\end{theorem}

\begin{proof}
It follows from explicit constructions provided in Appendix, see formulas (%
\ref{candcon}).$\square $
\end{proof}

\vskip5pt

We have the canonical distortion relation
\begin{equation}
\nabla =\widehat{\mathbf{D}}+\widehat{\mathbf{Z}},  \label{distorsrel}
\end{equation}%
when both linear connections $\nabla =\{\ _{\shortmid }\Gamma _{\ \beta
\gamma }^{\alpha }\}$ and $\widehat{\mathbf{D}}=\{\widehat{\mathbf{\Gamma }}%
_{\ \alpha \beta }^{\gamma }\}$ and the distorting tensor $\widehat{\mathbf{Z%
}}=\{\widehat{\mathbf{\ Z}}_{\ \alpha \beta }^{\gamma }\}$ are uniquely
defined by the same metric tensor $\mathbf{g}$ (\ref{fmetr}) (see
coefficient formulas (\ref{deflc}) and (\ref{deft}); for simplicity, we
omitted the left label ''$\mathbf{g}$''). The connection $\widehat{\mathbf{D}%
}$ is with nontrivial torsion (in general, the coefficients $\widehat{T}_{\
ja}^{i},\widehat{T}_{\ ji}^{a}$ and $\widehat{T}_{\ bi}^{a}$ are not zero,
see (\ref{dtors})).

\begin{definition}
The torsion $\widehat{\mathbf{T}}=\{\widehat{\mathbf{T}}_{\ \alpha \beta
}^{\gamma }\}$ of \ $\widehat{\mathbf{D}},$ completely defined by data $%
\left( \mathbf{g,N}\right) $ is \ called the canonical d--torsion of a
N--anholonomic (pseudo) Riemannian manifold $\mathbf{V.}$
\end{definition}

Such a torsion is nonholonomically induced by N--connection coefficients and
completely determined by certain off--diagonal N--terms in (\ref{fansatz}).
The GR theory can be formulated equivalently using the connection $\nabla $
and/or $\widehat{\mathbf{D}}$ if the distorting relation (\ref{distorsrel})
is used.

\subsection{The Einstein equations in N--adapted variables}

The Ricci d--tensor $Ric=\{\mathbf{R}_{\alpha \beta }\}$ of a d--connection $%
\mathbf{D}$ is constructed using a respective contracting of coefficients of
the curvature tensor (\ref{curv}), $\mathbf{R}_{\alpha \beta }\doteqdot
\mathbf{R}_{\ \alpha \beta \tau }^{\tau }.$ The h--/ v--components of this
d--tensor
\begin{equation}
\mathbf{R}_{\alpha \beta }=\{R_{ij}\doteqdot R_{\ ijk}^{k},\ \
R_{ia}\doteqdot -R_{\ ika}^{k},\ R_{ai}\doteqdot R_{\ aib}^{b},\
R_{ab}\doteqdot R_{\ abc}^{c}\},  \label{dricci}
\end{equation}%
see explicit coefficients formulas in Refs. \cite{vex3}. The scalar
curvature of $\mathbf{D}$ is constructed by using the inverse d--metric to $%
\mathbf{g}$ (\ref{dm}), $\mathbf{R}\doteqdot \mathbf{g}^{\alpha \beta }%
\mathbf{R}_{\alpha \beta }=g^{ij}R_{ij}+h^{ab}R_{ab},$ with $R=g^{ij}R_{ij}$
and $S=h^{ab}R_{ab}$ being respectively the h-- and v--components of scalar
curvature.

The Einstein equations for a metric $\mathbf{g}_{\beta \delta }$ are written
in standard form using the LC--connection $\nabla $ (and various tetradic,
spinor, 3+1 splitting etc representations). Using the distortion relations
of (\ref{distorsrel}), we can compute the respective distortions of the
Ricci tensor and of the scalar curvature and consider such values as
''effective'' sources additionally to the energy--momentum tensor for matter
fields, i.e. to $\varkappa T_{\beta \delta }$, where $\varkappa $ is defined
by the gravitational \ constant. If $\ ^{\mathbf{g}}\nabla =\ ^{\mathbf{g}}%
\mathbf{D}+\ ^{\mathbf{g}}\mathbf{Z,}$ we do not need additional field
equations (for torsion fields) like in the Einstein--Cartan, gauge or string
gravity theories.

\begin{theorem}
The Einstein equations in GR can be rewritten equivalently for a
nonholonomic 2+2 splitting and respective canonical d--connection $\widehat{%
\mathbf{D}},$%
\begin{eqnarray}
&&\widehat{\mathbf{R}}_{\ \beta \delta }-\frac{1}{2}\mathbf{g}_{\beta \delta
}\ ^{s}R=\mathbf{\Upsilon }_{\beta \delta },  \label{cdeinst} \\
&&\widehat{L}_{aj}^{c}=e_{a}(N_{j}^{c}),\ \widehat{C}_{jb}^{i}=0,\ \Omega
_{\ ji}^{a}=0,  \label{lcconstr}
\end{eqnarray}%
where the source $\mathbf{\Upsilon }_{\beta \delta }$ is such way
constructed that $\mathbf{\Upsilon }_{\beta \delta }\rightarrow \varkappa
T_{\beta \delta }$ \ for $\widehat{\mathbf{D}}\rightarrow \nabla .$
\end{theorem}

\begin{proof}
If the constraints (\ref{lcconstr}) are satisfied the tensors $\widehat{%
\mathbf{T}}_{\ \alpha \beta }^{\gamma }$ (\ref{dtors}) and $Z_{\ \alpha
\beta }^{\gamma }$(\ref{deft}) are zero. This states that $\widehat{\mathbf{%
\Gamma }}_{\ \alpha \beta }^{\gamma }=\Gamma _{\ \alpha \beta }^{\gamma },$
with respect to N--adapted frames, see (\ref{deflc}), even, in general, $%
\widehat{\mathbf{D}}\neq \nabla $ (the transformation laws under
frame/coordinate transforms of a d--connection and of the LC--connection are
different). In such a case, the equations (\ref{cdeinst}) are completely
equivalent to the Einstein equations for $\nabla .\ \square $
\end{proof}

\vskip5pt

For various purposes, it is more convenient to work with a d--connection $\
^{\mathbf{g}}\mathbf{D}$ instead of $\ ^{\mathbf{g}}\nabla .$ For instacne,
We can prescribe such a nonholonomic 2+2 splitting via N--coefficients $%
\mathbf{N=\{}N_{i}^{a}\}$ when $\ ^{\mathbf{g}}\mathbf{D}$ became a
canonical almost symplectic connection, which is important for deformation
quantization and/or A--brane quantization, or two--connection
renormalization of GR, see \cite{vdq3,vabr,vtwocon}. In section \ref{s6}, we
shall show that cancelations of anomalies are possible for $\widehat{\mathbf{%
D}}$ being a solution (\ref{cdeinst}) even, in general, quantum anomalies
exist for $\nabla .$

\subsection{On generic off--diagonal Einstein spaces}

It is possible a ''magic'' separation (i. e. decoupling; in our works, such
words are used as equivalent ones stating some properties of a system of
nonlinear partial differential equations) of the Einstein equations (\ref%
{cdeinst}) for $\widehat{\mathbf{D}}$ if certain types of \ parametrizations
for the coefficients of N--connection $\mathbf{N=\{}N_{i}^{a}\}$ of the
metric $\ \mathbf{g}$ (\ref{fmetr}) (equivalently, d--metric (\ref{dm})) are
considered, see details, proofs and examples in Refs. \cite{vex2,vex3,vp,vt}%
. This allows us to solve the Einstein equations in very general forms for
any general sources\footnote{%
such sources should be defined in explicit form from certain additional
suppositions on interactions of gravitational and matter fields; we omit
such considerations in this work} parametrized in the form
\begin{equation}
\mathbf{\Upsilon }_{\ \delta }^{\alpha }=diag[\mathbf{\Upsilon }%
_{2}(x^{k},y^{3}),\mathbf{\Upsilon }_{2}(x^{k},y^{3}),\mathbf{\Upsilon }%
_{4}(x^{k}),\mathbf{\Upsilon }_{4}(x^{k})],  \label{source}
\end{equation}
for coordinates $u^{\alpha }=(x^{k},y^{3}=v,y^{4}).$ In a particular case,
we can consider Einstein spaces with $\mathbf{\Upsilon }_{2}=\mathbf{%
\Upsilon }_{4}=\lambda = const$ which define off--diagonal Einstein spaces
with, in general, nontrivial cosmological constant $\lambda$.

\begin{lemma}
\label{lem} Any metric $\mathbf{g}$ (\ref{fmetr}) with coefficients of
necessary smooth class on a (pseudo) Riemannian $\mathbf{V}$ can be
represented in N--adapted form as
\begin{eqnarray}
\mathbf{g} &=&g_{i}(x^{k})dx^{i}\otimes dx^{i}+\omega
^{2}(x^{k},v,y^{4})h_{a}(x^{k},v)\mathbf{e}^{a}\otimes \mathbf{e}^{a},
\label{ans1} \\
\mathbf{e}^{3} &=&dy^{3}+w_{j}(x^{k},v)dx^{j},\ \mathbf{e}%
^{4}=dy^{4}+n_{j}(x^{k},v)dx^{j}.  \notag
\end{eqnarray}
\end{lemma}

\begin{proof}
A d--metric of type (\ref{ans1}) is of type (\ref{dm}) with $N_{i}^{3}=w_{j}$
and $N_{i}^{4}=n_{j}.$ With respect to a chosen coordinate base, we express (%
\ref{ans1}) in the form (\ref{fmetr}),{\small
\begin{equation}
g_{\underline{\alpha }\underline{\beta }}=\left[
\begin{array}{cccc}
g_{1}+\omega ^{2}(w_{1}^{\ 2}h_{3}+n_{1}^{\ 2}h_{4}) & \omega
^{2}(w_{1}w_{2}h_{3}+n_{1}n_{2}h_{4}) & \omega ^{2}w_{1}h_{3} & \omega
^{2}n_{1}h_{4} \\
\omega ^{2}(w_{1}w_{2}h_{3}+n_{1}n_{2}h_{4}) & g_{2}+\omega ^{2}(w_{2}^{\
2}h_{3}+n_{2}^{\ 2}h_{4}) & \omega ^{2}w_{2}h_{3} & \omega ^{2}n_{2}h_{4} \\
\omega ^{2}w_{1}h_{3} & \omega ^{2}w_{2}h_{3} & \omega ^{2}h_{3} & 0 \\
\omega ^{2}n_{1}h_{4} & \omega ^{2}n_{2}h_{4} & 0 & \omega ^{2}h_{4}%
\end{array}%
\right].  \label{ans1a}
\end{equation}%
} A general metric $g_{\alpha \beta }(x^{k},y^{a})$ on $\mathbf{V}$ can be
always parametrized in the form (\ref{ans1a}) via certain frame transform of
type $g_{\alpha \beta }=e_{\ \alpha }^{\underline{\alpha }}e_{\ \beta }^{%
\underline{\beta }}g_{\underline{\alpha }\underline{\beta }}$ (\ref%
{metrtransf}). For certain given values $g_{\alpha \beta }$ and $g_{%
\underline{\alpha }\underline{\beta }}$ (in GR, there are 6 + 6 independent
components), we have to solve a quadratic algebraic equation in order to
determine 16 coefficients $e_{\ \alpha }^{\underline{\alpha }},$ up to a
fixed coordinate system. We have to fix such nonholonomic 2+2 splitting when
the algebraic equations have real nondegenerate solutions. $\square $
\end{proof}

\vskip5pt

Let us introduce brief denotations for partial derivatives: \ $a^{\bullet
}=\partial a/\partial x^{1},$ $a^{\prime }=\partial a/\partial x^{2},a^{\ast
}=\partial a/\partial y^{3}$ and $a^{\diamond }=\partial a/\partial y^{4}$
and consider such parametrizations of (\ref{ans1}) when $h_{a}^{\ast }\neq 0$
(such conditions can be satisfied for some correspondingly chosen systems of
coordinates and frame transforms $e_{\ \alpha }^{\underline{\alpha }}).$
Using Lemma \ref{lem}, the solutions of the Einstein equations can be also
constructed in very general forms, see details and proofs in \cite{vex2,vex3}
(see also Appendix \ref{asec2}); we note that in different works there are
used some different parametrizations and/or systems of
reference/coordinates):

\begin{theorem}
\label{thgs}A general class of solutions of the gravitational field
equations (\ref{cdeinst}) for the canonical d--connection $\widehat{\mathbf{D%
}}$ with a general N--adapted diagonal source $\mathbf{\Upsilon }_{\ \delta
}^{\alpha }$ (\ref{source}), with $\mathbf{\Upsilon }_{2}\neq 0,$ is
determined by d--metrics of type (\ref{ans1}) (equivalently, by generic
off--diagonal metrics of type (\ref{ans1a})) with the coefficients computed
in the form%
\begin{eqnarray}
g_{i} &=&\epsilon _{i}e^{\psi (x^{k})},  \notag \\
h_{4} &=&\ \underline{h}_{4}(x^{k})\pm \int \left[ \Upsilon _{2}(x^{k},v)%
\right] ^{-1}\left( \exp [2\phi (x^{k},v)]\right) ^{\ast }dv,  \label{data1}
\\
h_{3} &=&\left[ \left( \sqrt{|h_{4}(x^{k},v)|}\right) ^{\ast }\right]
^{2}\exp [-2\phi (x^{k},v)],  \notag \\
w_{1} &=&\phi ^{\bullet }/\phi ^{\ast },w_{2}=\phi ^{\prime }/\phi ^{\ast },
\notag \\
n_{k} &=&\ ^{1}n_{k}(x^{i})+\ ^{2}n_{k}(x^{i})\int \left[ h_{3}/\left( \sqrt{%
|h_{4}|}\right) ^{3}dv\right]   \notag
\end{eqnarray}%
determined by an arbitrary generating function $\phi (x^{k},v),\phi ^{\ast
}\neq 0,$ any functions $\psi (x^{k})$ and $\omega (x^{k},y^{a})$ satisfying
the conditions
\begin{eqnarray}
\epsilon _{1}\psi ^{\bullet \bullet }+\epsilon _{2}\psi ^{^{\prime \prime }}
&=&\mathbf{\Upsilon }_{4},  \label{data2} \\
\mathbf{e}_{k}\omega =\partial _{k}\omega +w_{k}\omega ^{\ast }+n_{k}\omega
^{\diamond } &=&0,  \notag
\end{eqnarray}%
and integration functions $\ \underline{h}_{4}(x^{k}),$ $\
^{1}n_{k}(x^{i}),\ ^{2}n_{k}(x^{i})$ and $\ ^{0}h_{4}(x^{k})$ to be
determined from certain boundary conditions; in the above formulas $\epsilon
_{i}=\pm 1$ and other signs $\pm $ must be fixed such way to have a fixed
necessary signature for the chosen class of solutions.
\end{theorem}

The solutions defined by data (\ref{data1}) and (\ref{data2}) are very
general ones, with generic off--diagonal metrics when theirs coefficients
depend on all coordinates $u^{\alpha }=(x^{k},y^{a}).$ They may be of
arbitrary smooth class and with possible singularities, nontrivial
topological configurations which depends on the type of prescribed
symmetries (singularities etc) for generating and integration functions. For
such generic nonlinear systems, we can not state in general form any
uniqueness black hole theorems etc. Certain physical interpretation for some
classes of solutions can be provided only for special classes of
parametrization, see examples in section \ref{s5}.

We can prescribe such nonholonomic distributions with conventional 2+2
splitting when the general solutions from Theorem \ref{thgs} are restricted
to generate generic off--diagonal solutions in GR.

\begin{corollary}
\label{corol1} A d--metric (\ref{ans1}) with coefficients (\ref{data1}) and (%
\ref{data2}) define exact solutions of the Einstein equations for the
LC--connection $\nabla $ if the generating and integration functions are
constrained to satisfy the conditions%
\begin{equation}
w_{i}^{\ast }=\mathbf{e}_{i}\ln |h_{4}|,\mathbf{e}_{k}w_{i}=\mathbf{e}%
_{i}w_{k},n_{i}^{\ast }=0,\partial _{i}n_{k}=\partial _{k}n_{i}.
\label{lccond}
\end{equation}
\end{corollary}

\begin{proof}
By straightforward computations we can verify that if the constraints (\ref%
{lccond}) are solved the conditions (\ref{lcconstr}) are satisfied, i.e. the
torsion (\ref{dtors}) of $\widehat{\mathbf{D}}$ is zero. $\square $
\end{proof}

\vskip5pt

We note that if $\omega =const$ the generating solutions are with Killing
symmetry because the ansatz for metrics does not depend on variable $y^{4}.$
Such subclasses of off--diagonal solutions are also very general ones.

Finally, it should be emphasized that a similar \ theorem can be formulated
for vacuum solutions with $\mathbf{\Upsilon }_{\ \delta }^{\alpha }=0,$ see
details in Refs. \cite{vex2,vex3}. A limit $\mathbf{\Upsilon }_{\ \delta
}^{\alpha }\rightarrow 0$ may be a not smooth one on $\lambda $ for such
generic nonlinear solutions. For simplicity, in this work we shall consider
only Einstein spaces with nontrivial cosmological constant $\lambda ,$ $%
\mathbf{\Upsilon }_{\ \delta }^{\alpha }\rightarrow \lambda $ $\delta _{\
\delta }^{\alpha }$. \ The conditions (\ref{lccond}) have nontrivial
solutions, see Corollary \ref{corolap1}.

\section{Hidden Symmetries for Generic Off--Diagonal Solutions}

\label{s4} Following the conditions of Theorem \ref{thgs}, the Einstein
equations can be solved in very general forms if the coefficients of a
generic off--diagonal metric are defined by some data (\ref{data1}), (\ref%
{data2}) and (\ref{lccond}). In this section, we construct in explicit form
some classes of solutions describing nonolonomic deformations of the metrics
for Kerr--Sen black holes, analyze possible solitonic deformations and study
the associated hidden nonholonomic symmetries. Finally we shall study the
hidden symmetries of nonholonomic deformations of Kerr--Sen metrics.

\subsection{Nonholonomic analogous Kerr--Sen black holes in GR}

Let us consider an ansatz for a "primary" metric
\begin{eqnarray}
\ ^{\circ }\mathbf{g} &=&e^{\Phi }(\ ^{b}\rho )^{2}\left( \
^{0}S^{-1}dr\otimes dr+d\theta \otimes d\theta \right) +  \label{auxm1} \\
&&e^{\Phi }(\ ^{b}\rho )^{-2}\{a^{2}\sin ^{2}\theta \ e^{3}\otimes e^{3}-\
^{0}Se^{4}\otimes e^{4}\},  \notag \\
\mbox{ where } e^{3} &=&\delta t-\frac{r^{2}-2(M-b)r+a^{2}}{a^{2}}\delta
\varphi ,\ e^{4}=\delta t-a\sin ^{2}\theta \delta \varphi ,  \notag \\
\delta \varphi &=&d\varphi +w_{i}(u)dx^{i},\ \delta t=dt+n_{i}(u)dx^{i},
\notag
\end{eqnarray}%
with local coordinates $u^{\alpha }=(x^{1}=r,x^{2}=\theta ,y^{3}=\varphi
,y^{4}=t),$ for some given functions and constants
\begin{eqnarray*}
\widehat{\Phi } &=&\widehat{\Phi }(x^{i}),w_{k}=w_{k}(u^{\beta
}),n_{k}=n_{k}(u^{\beta }),\rho ^{2}=r^{2}+a^{2}\cos ^{2}\theta , \\
\ ^{b}\rho &=&\rho ^{2}(x^{i})+2br,\ ^{0}S(r)=r^{2}-2(M-b)r+a^{2}, \\
M &=&const,a=const,b=const.
\end{eqnarray*}%
In general, such a metric (\ref{auxm1}) is not a solution of the Einstein
equations. If we chose a respective class of coefficients for the metric
ansatz and matter fields,
\begin{eqnarray}
w_{i} &=&n_{i}=0,\widehat{\Phi }=2\ln (\rho /\ ^{b}\rho ),\ A=-rQ(\ ^{b}\rho
)^{-2}(dt-a\sin ^{2}\theta d\varphi ),  \label{data3} \\
H &=&2ab(\ ^{b}\rho )^{-4}d\varphi \wedge dt\wedge \left[ (r^{2}-a^{2}\cos
^{2}\theta )\sin ^{2}\theta dr-r\ ^{0}S\sin 2\theta d\theta \right] ,  \notag
\end{eqnarray}%
we generate the Kerr--Sen black hole solutions in the low--energy string
theory given by the effective action
\begin{equation}
I=\int \sqrt{|\ ^{\circ }\mathbf{g}|}d^{4}u\ e^{\Phi }\left( \ _{\mid
}^{\circ }R-\frac{1}{12}H_{\alpha \beta \gamma }H^{\alpha \beta \gamma }+\
^{\circ }\mathbf{g}^{\alpha \beta }\partial _{\alpha }\Phi \partial _{\beta
}\Phi -\frac{1}{8}F_{\alpha \beta }F^{\alpha \beta }\right) ,  \label{actks}
\end{equation}%
where $\ _{\mid }^{\circ }R$ is the Ricci scalar determined by the
LC--connection $\ ^{\circ }\nabla $ of $\ ^{\circ }\mathbf{g}_{\alpha \beta
},$ i.e. the metric in the string frame; $\widehat{\Phi }$ is the dilaton
field; $F=dA$ is the Maxwell field and $H=dB-\frac{1}{4}A\wedge dA$ is the
3--form for the skew--symmetric torsion field $H_{\alpha \beta \gamma }$
determined by an antisymmetric tensor field $B_{\alpha \beta }$ in string
gravity (see details in \cite{sen,blaga,wucai,houri}). The data (\ref{data3}%
) describe a black hole with mass $M$ and angular momentum $J=Ma,$ when a
nontrivial charge $Q$ induces a magnetic dipole \ momentum $\mu =Qa.$ When
the twist parameter $b=Q^{2}/2M$ is zero, we get the Kerr black hole known
in GR. In the so called Einstein frame metric when $\ ^{E}\mathbf{g}=e^{-%
\widehat{\Phi }}\ ^{\circ }\mathbf{g}$, the solutions present certain
fundamental string modifications of the Kerr geometry and possess some
inherit properties and hidden symmetries.

Our goal is to construct a new class of solutions in GR which will
generalize the Kerr metric as certain nonholonomic Kerr--Sen spacetimes when
some nontrivial N--connection coefficients $w_{k}=w_{k}(u^{\beta })$ and $%
n_{k}=n_{k}(u^{\beta })$ approximate off--diagonal gravitational
contributions of metrics which are similar to $H$ and $\widehat{\Phi }$
fields. Namely, we shall demonstrate that such solutions possess hidden
nonholonomic symmetries with canonical d--torsion occurring in the theory
for more general classes of frame deformations but vanishing if the
nonholonomic structure is correspondingly re--defined. For the reason to
understand the physical properties of nonholonomic modifications of
Kerr--Sen solutions for GR, compare with their analogs in the string frame
and show how geometric method can be applied for constructing exact
solutions in gravity theories, our study will be mainly concentrated on
ellipsoidal and solitonic deformation.

The ansatz (\ref{auxm1}) can be expressed equivalently in a diagonal form
\begin{equation}
\ ^{\circ }\mathbf{g}=\ ^{\circ }g_{1}\ dr\otimes dr+\ ^{\circ }g_{3}d\theta
\otimes d\theta +\ ^{\circ }h_{3}\ \ d\varphi \otimes d\varphi +\ ^{\circ
}h_{4}\ dt\otimes dt,  \label{auxm1a}
\end{equation}%
for a new system of coordinates
\begin{equation}
\widetilde{u}^{\beta }=(\widetilde{x}^{1}=\int \sqrt{|\ ^{0}S|}^{-1}(r)dr,%
\widetilde{x}^{2}=\theta ,\widetilde{y}^{3}=\varphi ,\widetilde{y}^{4}=t ),
\label{coordi}
\end{equation}%
where
\begin{eqnarray}
\ ^{\circ }g_{1}(r,\theta ) &=&\ ^{\circ }g_{2}(r,\theta )=e^{\Phi }(\
^{b}\rho )^{2},\ ^{\circ }h_{4}(r,\theta )=\ ^{1}K+\ ^{2}K,  \label{datai} \\
\ ^{\circ }h_{3}(r,\theta ) &=&\beta ^{2}(\ ^{2}K)(1+\ ^{2}K/\
^{1}Kq^{2})^{2}  \notag \\
\mbox{ for }\beta &=&-a\sin ^{2}\theta ,\ q=-(r^{2}+2br+a^{2})/a,  \notag \\
\ ^{1}K &=&-\ ^{0}S(\ ^{b}\rho )^{-2},\ \ ^{2}K=a^{2}\sin ^{2}\theta (\
^{b}\rho )^{-2}e^{\Phi }.  \notag
\end{eqnarray}

At the next step, we consider a nonholonomic deformation of the above metric
\begin{eqnarray}
\ ^{\circ }\mathbf{g} &= &[\ ^{\circ }g_{i},\ ^{\circ }h_{a},\ ^{\circ
}N_{i}^{a\,}=0]\rightarrow  \label{nhdef} \\
\ ^{\eta }\mathbf{g} &\mathbf{=}&[\ ^{\eta }g_{i}=\eta _{i}\ ^{\circ
}g_{i},\ ^{\eta }h_{a}=\eta _{a}\ ^{\circ }h_{a},\ \ ^{\eta
}N_{i}^{3\,}=w_{i},\ \ ^{\eta }N_{i}^{4\,}=n_{i}]  \notag
\end{eqnarray}%
for corresponding gravitational ''polarization'' functions
\begin{equation}
\eta _{\alpha }=1+\chi _{\alpha }=(\eta _{i}=1+\chi _{i}(x^{k}),\eta
_{a}=1+\chi _{a}(x^{k},\varphi ))  \label{polf}
\end{equation}%
and off--diagonal (N--connection) coefficients $\widetilde{w}%
_{i}(x^{k},\varphi ))$ and $\widetilde{n}_{i}(x^{k},\varphi )),$ when
\begin{equation*}
\ d\varphi \rightarrow \mathbf{e}^{3}=d\varphi +\widetilde{w}_{i}d\widetilde{%
x}^{i},\ dt\rightarrow \mathbf{e}^{4}=dt+\widetilde{n}_{i}d\widetilde{x}^{i}.
\end{equation*}%
Such a ''target'' d--metric $\ ^{\eta }\mathbf{g}$ is supposed to be generic
off--diagonal of type (\ref{ans1}) (equivalently (\ref{ans1a})), with
respective coefficients chosen in the forms (\ref{data1}) and (\ref{data2}),
and must define exact solutions of the Einstein equations for the canonical
d--connection and, for the corresponding restrictions, for the
LC--connection, when the source $\mathbf{\Upsilon }_{\ \delta }^{\alpha
}\rightarrow \lambda $ $\delta _{\ \delta }^{\alpha },$ see formula (\ref%
{source}) and the conditions stated in Lemma \ref{lem}, Theorem \ref{thgs}
and Corollary \ref{corol1}.

For some classes of nonholonomic deformations, we can consider that some $%
\chi _{\alpha }$ in polarizations (\ref{polf}) are small values, $|\chi
_{\alpha }|<1,$ which allows us to provide certain physical interpretation
of such classes of solutions to be very similar to that for the Kerr--Sen
metrics but (in our case) in GR, with some off--diagonal modifications, let
say, with ellipsoid and/or solitonic symmetries.

Using the results provided in Corollary \ref{corolap1} and Remark \ref%
{remarka}, we prove

\begin{theorem}
\label{theorsmp}The set of stationary solutions for (\ref{nhdef}) defining
nonholonomic deformations of a originating from string gravity d--metric (%
\ref{auxm1}) into generic off--diagonal metrics for Einstein spacetimes in
GR are parametrized by gravitational polarizations
\begin{eqnarray}
\eta _{i} &=&1+\chi _{i}=e^{\psi (\widetilde{x}^{i})},\eta _{i}=1+\chi _{a}
\label{defsol} \\
w_{i} &=&(\phi ^{\ast })^{-1}\frac{\partial \phi }{\partial \widetilde{x}^{i}%
},\ n_{k}=\ ^{1}n_{k}(\widetilde{x}^{i}),  \notag
\end{eqnarray}%
generated by a function $\chi _{4}(\widetilde{x}^{i},\varphi ),$ for $\chi
_{4}^{\ast }\neq 0,$ where $\phi =\ln \sqrt{|\lambda \ ^{\circ }h_{4}\chi
_{4}|}$ and $\chi _{3}=-1+(\lambda \ ^{\circ }h_{3})^{-1}\left[ \left( \ln
\sqrt{|1-\chi _{4}|}\right) ^{\ast }\right] ^{2}$ (\ref{aux4a}) with $\psi (%
\widetilde{x}^{i})$ being a solution \newline
of $[\frac{\partial ^{2}}{(\partial \widetilde{x}^{1})^{2}}+\frac{\partial
^{2}}{(\partial \widetilde{x}^{2})^{2}}][\psi (\widehat{\Phi }+2\ln |\
^{b}\rho |)]=\lambda $ and $\ \partial (\ ^{1}n_{k})/\partial \widetilde{x}%
^{i}=\partial (\ ^{1}n_{i})/\partial \widetilde{x}^{k}.$
\end{theorem}

The set of solutions (\ref{defsol}) is defined for any $\lambda \neq 0$ (in
a similar form, it is possible to generate vacuum solutions, see details in
\cite{vex2,vex3}; the lengths of this paper does not allow us to consider
such metrics).

\subsection{Ellipsoidal deformations and gravitational solitons}

\label{sssolit}

\subsubsection{Rotoid deformations}

We can chose a class of noholonomic deformations when data (\ref{defsol})
define a black ellipsoid solution in GR (such solutions were studied also in
(non) commutative and/or string/brane models of gravity; they seem to be
stable and, for Einstein configurations, do not violate the conditions of
black hole uniqueness theorems, see details in \cite{vsingl1,vex2,vex3} and
references therein).

The generating function (gravitational polarization) is chosen $\chi
_{4}=-1+(\ ^{\circ }h_{4})^{-1}\left( \underline{q}+\varepsilon \underline{%
\varrho }\right) $, where
\begin{equation}
\underline{q}(\widetilde{x}^{1},\theta ,\varphi )=1-\frac{2\ ^{1}\underline{%
\mu }(r,\theta ,\varphi )}{r},\ \ \underline{\varrho }(\widetilde{x}%
^{1},\theta ,\varphi )=\frac{\underline{q}_{0}(r)}{4\underline{\mu }_{0}^{2}}%
\sin (\omega _{0}\varphi +\varphi _{0}).  \label{ellipsv}
\end{equation}%
Following Theorem \ \ref{theorsmp}, there is a class of exact solutions (for
any fixed parameter $\varepsilon $)
\begin{eqnarray}
~_{\lambda }^{rot}\mathbf{g} &=&e^{\psi (\widetilde{x}^{1},\theta )}\
^{\circ }g_{1}\left( d\widetilde{x}^{1}\otimes d\widetilde{x}^{1}+\ d\theta
\otimes d\theta \right)+  \notag \\
&&\lambda ^{-1}\left[ \left( \ln \sqrt{|2-\underline{q}-\varepsilon
\underline{\varrho }|}\right) ^{\ast }\right] ^{2}\delta \varphi \otimes \
\delta \varphi -\left( \underline{q}+\varepsilon \underline{\varrho }\right)
\ \delta \widetilde{t}\otimes \delta \widetilde{t},  \notag \\
\delta \varphi &=&d\varphi +w_{1}d\widetilde{x}^{1}+w_{2}d\theta ,\ \delta
t=dt+n_{1}d\widetilde{x}^{1}+n_{2}d\theta .  \label{soladel}
\end{eqnarray}%
Such generic off--diagonal (anisotropic) stationary metrics posses rotoid
symmetry with ''ellipsoidal horizon'' when the condition of vanishing of the
metric coefficient before $\delta t\otimes \delta t,$ i.e. $\ h_{4}=0,$
states a parametric elliptic configuration $\ $of type $r_{+}\simeq 2\ ^{1}%
\underline{\mu }/\left( 1+\varepsilon \frac{\underline{q}_{0}(r)}{4%
\underline{\mu }_{0}^{2}}\sin (\omega _{0}\varphi +\varphi _{0})\right) ,$
for a corresponding $\underline{q}_{0}(r).$ The parameter $\varepsilon $ is
just the eccentricity for a rotating ellipsoid.

For small values of $\varepsilon \geq 0,$ a metric (\ref{soladel}) describes
nonholonomic black ellipsoid -- de Sitter configurations, when the
coefficients of metrics \ contain dependencies on ''primary'' values like $\
^{\circ }g_{1}$ and $\ ^{\circ }h_{4}$ in (\ref{auxm1}) (equivalently, (\ref%
{auxm1a})). We argue that such solutions define in GR certain rotoid black
hole objects mimicking ''nonholonomically deformed'' black holes. The
physical properties of such solutions are determined by three sets of data:
the ''primary'' data for the Kerr--Sen metrics, the ellipsoidal type of
nonholonomic deformations stated by values (\ref{ellipsv}) \ for $\chi _{4},$
and a nontrivial cosmological constant $\lambda .$

In the limit $\varepsilon \rightarrow 0,$ we get a subclass of solutions
with spherical symmetry but with generic off--diagonal coefficients induced
by the N--connection coefficients, i.e. by the corresponding nonholonomic
deformations. This class of spacetimes depend on cosmological constants
which, in general, are polarized nonholonomically by nonlinear gravitational
interactions. We can extract from such configurations the Schwarzschild
solution if we select a set of functions with the properties $\phi
\rightarrow const,w_{i}\rightarrow 0,n_{i}\rightarrow 0$ and $%
h_{4}\rightarrow \varpi ^{2},$ where $\varpi ^{2}$ is chosen to determine
the horizon of a static black hole. Finally, we emphasize that the
parametric dependence on cosmological constants, for such nonholonomic
configurations, is not smooth.

\subsubsection{Kerr--Sen rotoids and solitonic distributions}

The solutions for the Kerr--Sen rotoid configurations in GR $\ ^{rot}\mathbf{%
g}$ (\ref{soladel}) can be generalized by introducing additional stationary
deformations induced by a static three dimensional solitonic distribution $%
\eta (\widetilde{x}^{1},\theta ,\varphi )$ as a solution of the solitonic
equation\footnote{$\eta $ can be a solution of any three dimensional
gravitational solitonic stationary solitonic distribution and/ or other
nonlinear wave equations if we construct exact solutions with running in
time solitons when are functions of type $\eta (t,\theta ,\varphi )$, or $%
\eta (\xi ,\theta ,t)$, see \cite{vsingl1}}
\begin{equation}
\eta ^{\bullet \bullet }+\epsilon (\eta ^{\prime }+6\eta \ \eta ^{\ast
}+\eta ^{\ast \ast \ast })^{\ast }=0,\ \epsilon =\pm 1.  \label{sol3d}
\end{equation}%
We construct and analyze two different types of solitonic nonholonomic
deformations of Kerr--Sen black holes.

\paragraph{De Sitter type solutions with 3--d solitonic polarizations of
masses:\ }

The geometric ''target'' data are generated by a ''vertical'' distribution $%
\eta (\widetilde{x}^{1},\theta ,\varphi )$ when the coefficients of
v--metric are computed in the form
\begin{equation}
h_{4}=\eta ^{\bullet \bullet }=\left( \underline{q}+\varepsilon \underline{%
\varrho }\right) ,\ h_{3}=\left[ \left( \sqrt{|h_{4}|}\right) ^{\ast }\right]
^{2}e^{-2\phi },  \notag
\end{equation}%
where \
\begin{equation}
\phi (\widetilde{x}^{1},\theta ,\varphi )=\frac{1}{2}\ln \left| \lambda %
\left[ \ ^{0}h_{4}+\epsilon \left( \eta ^{\prime }+6\eta \eta ^{\ast }+\eta
^{\ast \ast \ast }\right) ^{\ast }\right] \right| .  \label{solgfunct}
\end{equation}%
By straightforward computations, we can verify that the conditions of
Theorem \ \ref{theorsmp} with $h_{4}=\ \ ^{\circ }h_{4}\pm \lambda
^{-1}e^{2\phi }$ are satisfied if and only if $\eta $ is a solution of (\ref%
{sol3d}).

Using formula (\ref{ellipsv}), we can relate the solitonic function $\eta $
to certain gravitational polarizations of mass, when (for simplicity, we can
consider $\varepsilon =0)$
\begin{equation}
\eta ^{\bullet \bullet }(\widetilde{x}^{1},\theta ,\varphi )=1-\frac{2\ ^{1}%
\underline{\mu }(r,\theta ,\varphi )}{r}.  \label{solitpol}
\end{equation}%
Putting together the coefficients, we get the metric
\begin{eqnarray}
~_{1sol}^{rot}\mathbf{g} &=&e^{\psi (\widetilde{x}^{1},\theta )}\ ^{\circ
}g_{1}\left( d\widetilde{x}^{1}\otimes d\widetilde{x}^{1}+\ d\theta \otimes
d\theta \right)   \label{solbes} \\
&&+\left[ \left( \sqrt{|\eta ^{\bullet \bullet }|}\right) ^{\ast }\right]
^{2}e^{-2\phi }\delta \varphi \otimes \ \delta \varphi -\eta ^{\bullet
\bullet }\ \ \delta t\otimes \delta t,  \notag \\
\delta \varphi  &=&d\varphi +w_{1}(\widetilde{x}^{1},\theta ,\varphi )d%
\widetilde{x}^{1}+w_{2}(\widetilde{x}^{1},\theta ,\varphi )d\theta ,\
\notag \\
\delta t &=&dt+\ ^{1}n_{1}(\widetilde{x}^{1},\theta )d\widetilde{x}^{1}+\
^{1}n_{2}(\widetilde{x}^{1},\theta )d\theta ,  \notag
\end{eqnarray}%
where the N--connection coefficients $w_{i}=(\phi ^{\ast })^{-1}\frac{%
\partial \phi }{\partial \widetilde{x}^{i}}$ and $n_{k}=\ ^{1}n_{k}(%
\widetilde{x}^{i})$ are respectively computed using the generating function $%
\phi $ (\ref{solgfunct}) and subjected to the LC--conditions (\ref{lccond}).
$\ $

The class of generic off--diagonal metrics (\ref{solbes}) \ are similar to
the rotoid Kerr--Sen configurations in GR (\ref{soladel}). For small $%
\varepsilon ,$ such solutions define black ellipsoids with solitonically
polarized mass $\ ^{1}\underline{\mu }(r,\theta ,\varphi )=\mu
_{0}+\varepsilon \mu _{1}(r,\theta ,\varphi ),$ following (\ref{solitpol}).
We suppose that it is possible \ to detect such black rotoid objects via
anisotropic polarizations of their masses when certain gravitational
solitonic waves act on such generic off--diagonal black ellipsoid solutions.

\paragraph{Arbitrary solitonic deformations of the v--components of
d--metrics:}

We construct another class of nonholonomic solitonic transforms from $\
^{rot}\mathbf{g}$ (\ref{soladel}) to a stationary $~_{2st}^{rot}\mathbf{g}$
determining stationary metrics for a rotoid in solitonic backgrounds,
\begin{eqnarray}
~_{2st}^{rot}\mathbf{g} &=&e^{\psi (\widetilde{x}^{1},\theta )}\ ^{\circ
}g_{1}\left( d\widetilde{x}^{1}\otimes d\widetilde{x}^{1}+\ d\theta \otimes
d\theta \right) +  \label{solrot} \\
&&\left[ \left( \sqrt{|\eta \left( \underline{q}+\varepsilon \underline{%
\varrho }\right) |}\right) ^{\ast }\right] ^{2}\ \delta \varphi \otimes \
\delta \varphi -\eta \left( \underline{q}+\varepsilon \underline{\varrho }%
\right) \ \delta t\otimes \delta t,  \notag \\
\delta \varphi  &=&d\varphi +w_{1}(\widetilde{x}^{1},\theta ,\varphi )d%
\widetilde{x}^{1}+w_{2}(\widetilde{x}^{1},\theta ,\varphi )d\theta ,\
\notag \\
\delta t &=&dt+\ ^{1}n_{1}(\widetilde{x}^{1},\theta )d\widetilde{x}^{1}+\
^{1}n_{2}(\widetilde{x}^{1},\theta )d\theta .  \notag
\end{eqnarray}%
The conditions of Theorem \ref{theorsmp} and Remark \ref{remarka} are
satisfied for any nontrivial solution $\eta (\widetilde{x}^{1},\theta
,\varphi )$ of the solitonic equation (\ref{sol3d}) if we chose the
generating function
\begin{equation}
\phi =\frac{1}{2}\ln \left\vert \lambda \left[ \ ^{0}h_{4}\pm \eta \left(
\underline{q}+\varepsilon \underline{\varrho }\right) \right] \right\vert
\label{genfunct3}
\end{equation}%
for computing the values $w_{i}=(\phi ^{\ast })^{-1}\frac{\partial \phi }{%
\partial \widetilde{x}^{i}}$ and $n_{k}=\ ^{1}n_{k}(\widetilde{x}^{i})$
subjected to the LC--conditions (\ref{lccond}).

There is a substantial difference between two classes of solitonic
gravitational solutions (\ref{solbes}) and (\ref{solrot}). In the first
case, the 3-d distribution $\eta (\widetilde{x}^{1},\theta ,\varphi )$ is a
solitonic one if the Einstein equations are satisfied. In the second class,
the distribution may be a solitonic type, or another one, and such
off--diagonal metric ansatz defining Einstein spacetime manifolds do not
transform the field equations into certain solitonic ones. The generating
function $\phi (\widetilde{x}^{1},\theta ,\varphi )$ (\ref{genfunct3}) is
such way chosen that $\ _{2st}^{rot}\mathbf{g}$ (\ref{solrot}) describes a
rotoid Kerr--Sen configuration $\ ^{rot}\mathbf{g}$ (\ref{soladel}) imbedded
self--consistently into a solitonic background $\eta (\widetilde{x}%
^{1},\theta ,\varphi )$ \thinspace (\ref{sol3d}) in such forms that, for
instance, black ellipsoid properties are preserved for various types of
nonlinear waves.

\subsection{Hidden symmetries for nonholonomic Kerr--Sen metrics}

The quadratic elements for all classes of solitonic Kerr--Sen rotoid
configurations constructed in this section can be represented as
nonholonomic deformations (\ref{nhdef}) of the metric (\ref{auxm1a}) via
gravitational polarizations (\ref{polf}),%
\begin{eqnarray*}
\ \delta s^{2} &=&\ ^{\circ }g_{i}(r,\theta )\ (\sqrt{|\eta _{i}(r,\theta )|}%
d\widetilde{x}^{i})^{2}+\ ^{\circ }h_{a}(r,\theta )\ \ (\sqrt{|\eta
_{a}(r,\theta ,\varphi )|}\widetilde{\mathbf{e}}^{a})^{2}, \\
\widetilde{\mathbf{e}}^{3} &=&d\varphi +\widetilde{w}_{i}(r,\theta ,\varphi
)d\widetilde{x}^{i},\ \widetilde{\mathbf{e}}^{4}=d\varphi +\widetilde{n}%
_{i}(r,\theta ,\varphi )d\widetilde{x}^{i},
\end{eqnarray*}%
for $\widetilde{x}^{i}=(\widetilde{x}^{1}(r),\theta )$ and $\widetilde{y}%
^{a}=y^{a}=(\varphi ,t)$. The coefficients/multiples in this formula are
considered in a form when an exact solution for the Kerr--Sen and/or
ellipsoidal/solitonic configuration is defined.

Let us introduce the following basis of 1--forms:{\small
\begin{eqnarray}
\overline{\mathbf{e}}^{\mu } &=&[\overline{e}^{1}=\rho \sqrt{|\eta _{1}|}d%
\widetilde{x}^{1},\overline{e}^{2}=\sqrt{|\eta _{2}|}\rho d\theta ,\
\overline{\mathbf{e}}^{3} = (\ ^{b}\rho )^{-2}\rho a\sin \theta (\sqrt{|\eta
_{4}|}\widetilde{\mathbf{e}}^{4}+q\sqrt{|\eta _{3}|}\widetilde{\mathbf{e}}%
^{3}),  \notag \\
\overline{\mathbf{e}}^{4} &=&(\ ^{b}\rho )^{-2}\rho \sqrt{|\ ^{0}S|} (\sqrt{%
|\eta _{4}|}\widetilde{\mathbf{e}}^{4}+\beta \sqrt{|\eta _{3}|}\widetilde{%
\mathbf{e}}^{3})]  \label{dualbas}
\end{eqnarray}%
} when the target metric is written
\begin{equation}
\ ^{\eta }\mathbf{g}\mathbf{=}\overline{e}^{1}\otimes \overline{e}^{1}+%
\overline{e}^{2}\otimes \overline{e}^{2}+\overline{\mathbf{e}}^{3}\otimes
\overline{\mathbf{e}}^{3}-\overline{\mathbf{e}}^{4}\otimes \overline{\mathbf{%
e}}^{4}.  \label{dm1}
\end{equation}%
In general, the nonholonomically induced canonical d--torsion $\widehat{%
\mathcal{T}}$ (\ref{dtors}) is not completely antisymmetric. Nevertheless,
alternatively to this d--torsion, we can define another torsion field (which
is also completely determined nonholonomically by the metric structure and
parametrized by two functions before absolute anti--symmetric tensors, $%
e^{234}$ and $e^{124}$, with respect to a dual base (\ref{dualbas})),%
\begin{eqnarray*}
\ ^{\triangleleft }T^{\pm } &=&-2a\sin \theta \left[ -(\ ^{b}\rho
)^{-2}(r+b)\pm r\rho ^{-2}\right] , \\
\ ^{\triangleright }T^{\pm } &=&-2\rho ^{-1}a(\cos \theta )\sqrt{|\ ^{0}S|}%
\left[ -(\ ^{b}\rho )^{-2}\pm \rho ^{-2}\right] ,
\end{eqnarray*}%
when
\begin{equation}
\ ^{H}\mathbf{T^{\pm }:}=\ ^{\triangleright }T^{\pm }e^{234}+\
^{\triangleleft }T^{\pm }e^{124}.  \label{auxtors}
\end{equation}%
The torsions $\ ^{H}\mathbf{T}^{\pm }$ are similar to the torsions
associated to $H$ studied in Ref. \cite{houri}. Nevertheless, constructing
our solutions we have not involved additional field equations for the matter
fields and torsion (we shall write in brief $\ ^{H}\mathbf{T}$ if it will
not be not necessary to state exactly with what sign $+,$ or $-,$ we are
working for some constructions). There is not a general smooth transform $\
^{H}\mathbf{T\rightarrow }H$ because our gravitational field equations for
target metrics are not derived from an action of type (\ref{actks}). For
some classes of nonholonomic constraints on the off--diagonal gravitational
field dynamics we can mimic certain types of matter field interactions.

There are two obvious isometries $\partial _{t}$ and $\partial _{\varphi }$
for a general (non) holonomic Kerr--Sen metric. Such a geometry also admits
an irreducible Killing tensor\footnote{%
such a geometric object is responsible for separability of (charged)
Hamilton--Jacobi equation in the complete integrability of the motion of
particles; with respect to anholonomic frames, the constructions can be
generalized to include nonholonic variables and non--integrable dynamics}%
\begin{equation*}
\ ^{\eta }\overline{\mathbf{K}}=a^{2}\cos ^{2}\theta \left( \overline{%
\mathbf{e}}^{4}\overline{\mathbf{e}}^{4}-\overline{e}^{1}\overline{e}%
^{1}\right) +r^{2}(\overline{e}^{2}\otimes \overline{e}^{2}+\overline{%
\mathbf{e}}^{3}\otimes \overline{\mathbf{e}}^{3}).
\end{equation*}%
The first one can be related to the canonical d--torsion $\widehat{\mathcal{T%
}}$ (\ref{dtors}) and/or auxiliary torsion $\ ^{H}\mathbf{T}$ (\ref{auxtors}%
) in such a form when a nonholonomically induced torsion is included
naturally into a generalized closed CKY 2--form. We change $\widehat{\mathbf{%
T}}\rightarrow \ ^{H}\mathbf{T}^{+}$ and consider $\ \ ^{H}\mathbf{D}_{%
\mathbf{X}}\Psi :=\nabla _{\mathbf{X}}\Psi +\frac{1}{2}(\mathbf{X\rfloor }\
^{H}\mathbf{T}^{+})\bigwedge\limits_{1}\Psi ,$ when {\small
\begin{eqnarray*}
\ ^{H}\mathbf{d}\Psi &=&\mathbf{d}\Psi -\ ^{H}\mathbf{T}^{+}\bigwedge%
\limits_{1}\Psi =\overline{\mathbf{e}}^{\alpha }\wedge \ ^{H}\mathbf{D}%
_{\alpha }\Psi =\overline{e}^{i}\wedge \ ^{H}\mathbf{D}_{i}\Psi +\overline{%
\mathbf{e}}^{a}\wedge \ ^{H}\mathbf{D}_{a}\Psi, \\
\ ^{H}\mathbf{d}^{\ast }\Psi &=&\mathbf{d}^{\ast }\Psi -\ ^{H}\mathbf{T}%
^{+}\bigwedge\limits_{2}\Psi =-\overline{\mathbf{e}}_{\alpha }\rfloor \ ^{H}%
\mathbf{D}_{\alpha }\Psi =-\overline{\mathbf{e}}_{i}\rfloor\ ^{H}\mathbf{D}%
_{i}\Psi -\overline{\mathbf{e}}_{a}\rfloor\ ^{H}\mathbf{D}_{a}\Psi.
\end{eqnarray*}
}

We can verify explicitly that:

\begin{corollary}
On $(n+m)$--dimensional N--anholonomic (pseudo) Riemannian space, the value $%
\ ^{+}h=a\cos \theta \overline{\mathbf{e}}^{3}\wedge \overline{e}^{2}+r%
\overline{\mathbf{e}}^{4}\wedge \overline{e}^{1}$ is a closed nonholonomc
CKY 2--form obeying the conditions%
\begin{equation*}
\ ^{H}\mathbf{D}_{\mathbf{X}}(\ ^{+}h)=\mathbf{X}^{\flat }\wedge \ \ ^{H}%
\mathbf{\xi },\ \ ^{H}\mathbf{\xi =}(1-n-m)^{-1}\ ^{H}\mathbf{d}^{\ast }(\
^{+}h).
\end{equation*}
\end{corollary}

If both conditions $\widehat{\mathbf{Z}}=0$ (\ref{lcconstr}) and $\ ^{H}%
\mathbf{T}^{+}=0,$ i.e. for holonomic configurations, the existence of $\
^{+}h$ determines explicitly a tower of hidden symmetries and states
uniquely (up to $(n+m)/2$ functions of one variable) the canonical form of
metric. Such holonomic principal CKY tensors and related symmetries are
studied in Refs. \cite{krt1,krt2}. For generic off--diagonal metrics, hidden
nonholonomic symmetries exist naturally and they are associated with induced
torsions. The constructions are similar to those presented in sections 2 and
3 of Ref. \cite{houri} for generalized closed CKY 2-forms with nontrivial
torsion and generated towers of hidden symmetries with that difference that
in our work we consider different types of torsions.

A nontrivial nonholonomic structure does not allow a simple recovering of
symmetries from the divergence of $\ ^{+}h$ (contrary to the ''holonomic''
Kerr solution). Following the conditions of above Corollary, we can
introduce such nonholonomic values
\begin{eqnarray*}
\ ^{H}\mathbf{\xi }_{+} &=&- \frac{1}{3}\ ^{H}\mathbf{d}^{\ast }(\
^{+}h)=\rho ^{-1} ( -\sqrt{|\ ^{0}S|}\overline{\mathbf{e}}^{4}+a\sin \theta
\ \overline{\mathbf{e}}^{3}), \\
\ ^{H}\mathbf{\xi }_{+}^{\natural } &=& ( \sqrt{|\eta _{4}|}\ ^{\circ
}h_{4})^{-1}\partial _{t}.
\end{eqnarray*}
For holonomic configurations, there is a smooth limit $\eta _{4}\rightarrow
0 $ to data (\ref{data3}) with $b=0$ in (\ref{auxm1}), when both torsions $%
\widehat{\mathcal{T}}$ and $\ ^{H}\mathbf{T}$ vanish and we recover the
standard Kerr geometry.

There is another closed nonholonomic CKY 2--form associated to nonholonomic
and holonomic Kerr--Sen type geometries
\begin{eqnarray*}
\ ^{-}h &=&a\cos \theta \ \overline{e}^{2}\wedge \ \overline{\mathbf{e}}%
^{3}-r\ \overline{e}^{1}\wedge \ \overline{\mathbf{e}}^{4}, \\
\ ^{H}\mathbf{\xi }_{-} &\mathbf{=}&\mathbf{-}\frac{1}{3}\ ^{H}\mathbf{d}%
^{\ast }(\ ^{-}h)=-\rho ^{-1}\left( \sqrt{|\ ^{0}S|}\overline{\mathbf{e}}%
^{4}+a\sin \theta \ \overline{\mathbf{e}}^{3}\right) ,
\end{eqnarray*}%
with respect to a different torsion $\ ^{H}\mathbf{T}^{-}:=\
^{\triangleright }T^{-}e^{234}+ \ ^{\triangleleft}T^{-}e^{124}$. In Ref.
\cite{houri}, see formulas (3.15) - (3.17), it is considered that such a
torsion is ''rather peculiar'' because it remains non--trivial even in the
holonomic Kerr geometry. In terms of the geometry of nonholonomic manifolds,
this is not surprising because rotating Kerr black holes can be described
naturally with respect to rotating system of coordinates. Such
geometric/physical objects have already a specific anisotropy determined by
rotations and the associated torsion $\ ^{H}\mathbf{T^{-}}$ are induced by
rotating frames of reference. In a more general context, the constructions
can be extended to arbitrary nonholonomic frames including those related to
N--connections and certain off--diagonal terms in the metrics.

Finally, in this section, we note the nontrivial nonholonomic structures
associated to off--diagonal solutions are characterized additionally by
various types of hidden nonholonomic symmetries with induced torsions which
play an important role in stating criteria of separability of
Hamilton--Jacobi, Klein--Gordon, Dirac equations etc, which nonholonomic
variables. Such constructions were provided, for instance, in Refs. \cite%
{houri}, for holonomic Kerr--Sen and Taub NUT backgrounds, and in Refs. \cite%
{vp,vt}, for nonholonomic Taub NUT solutions and Dirac operators. The length
of this paper does not allow analyze certain constructions related to
nonholonomic Kerr--Sen configurations.

\section{Nonholonomic Symmetries and Quantum Gravitational Ano\-ma\-lies}

\label{s5} In previous section, we proved that there are natural
nonholonomic hidden symmetries determined by non--diagonal components of
metrics which characterize different classes of generic off--diagonal
solutions of Einstein equations. Usual hidden symmetries, for holonomic
configurations, with the LC--connection, result in gravitational anomalies
(studied in various models of quantum gravity). The surprising thing is that
we can find such nonholonomic hidden symmetries when certain classes of
anomalies can be canceled (for some auxiliary, but not less fundamental,
connections); at the end, all constructions can be re--defined equivalently
in terms of the LC--configurations. The aim of this section is to show how
some types of quantum gravitational anomalies derived for nonholonomic
hidden symmetries can be eliminated.

We consider a physical system when possible non--integrable constraints are
encoded into the frame structure of a nonholonomic manifold $\mathbf{V,}\dim
\mathbf{V}=n+m$. To find the necessary conditions for the existence of
constants of, in general, constrained motion in a first-quantized system we
replace momenta by derivatives and look for operators commuting with the
Hamiltonian $\widehat{\mathcal{H}}$. In our approach, it is defined using
the canonical d--connection,
\begin{equation}
\widehat{\mathcal{H}}=\widehat{\square }=\widehat{\mathbf{D}}_{\beta }%
\mathbf{g}^{\beta \gamma }\widehat{\mathbf{D}}_{\gamma }=\widehat{\mathbf{D}}%
_{\beta }\widehat{\mathbf{D}}^{\beta }.  \label{KG}
\end{equation}%
For holonomic configurations, $\widehat{\mathcal{H}}$ transforms into the
Hamiltonian $\mathcal{H}$ defined by $\nabla $ and corresponds to a free
scalar particle when the covariant Laplacian (d'Alambert operator $\square
=\nabla _{\beta }\nabla ^{\beta },$ constructed for the same metric) is
acting on scalars.

In general, the classical conserved quantities associated with SK tensors do
not transfer to the quantized systems which results in quantum anomalies. In
what follows we shall analyze the quantum anomalies in the case of
nonholonomic SK and CSK d--tensors and, to make things more specific, we
confine ourselves to the case of tensors of rank $1$ and $2$.

Let us consider a conserved operator corresponding to a conformal Killing
d--vector $\mathbf{K}^{\alpha }$ in the quantized system $\ _{V}\widehat{%
\mathcal{Q}}=\mathbf{K}^{\alpha }\widehat{\mathbf{D}}_{\alpha }$. A quantum
gravitational anomaly can be identified by evaluating the commutator $[%
\widehat{\square },\widehat{\mathcal{Q}}_{V}]\Phi $ for the solutions of the
Klein-Gordon equation, $\Phi \in \mathcal{C}^{\infty }(\mathbf{V})$, with
the Klein-Gordon d--operator (\ref{KG}). A straightforward calculation with
respect to N--adapted frames is similar to that in usual (pseudo) Riemannian
spaces but using $\widehat{\mathbf{D}}$ instead of $\nabla$. It gives
\begin{equation}
\lbrack \widehat{\mathcal{H}},\ \ _{V}\widehat{\mathcal{Q}}]=\frac{2-n-m}{n+m%
}\left( \widehat{\mathbf{D}}^{\alpha }\widehat{\mathbf{D}}^{\gamma }\mathbf{K%
}_{\gamma }\right) \widehat{\mathbf{D}}_{\alpha }+\frac{2}{n+m}\widehat{%
\mathbf{D}}_{\alpha }\mathbf{K}_{~}^{\alpha }\widehat{\square }\,.
\label{commck}
\end{equation}

For a Killing d--vector $\mathbf{K}^{\alpha },$ this commutator vanishes and
there are no quantum gravitational anomalies (this property holds both for
the holonomic and/or nonholonomic systems). But for conformal Killing
d--vectors the \ situation is very different. Even for solutions of the
massless nonholonomic Klein--Gordon equation, $\widehat{\square }\Phi (x)=0$
(we can also consider $\square \Phi (x)=0$, when the zero torsion conditions
are satisfied), the term $\left( \widehat{\mathbf{D}}^{\alpha }\widehat{%
\mathbf{D}}^{\gamma }\mathbf{K}_{\gamma }\right) \widehat{\mathbf{D}}%
_{\alpha }$ from (\ref{commck}) survives and, in general, the system is
affected by quantum gravitational anomalies.

Using a tedious evaluation for N--adapted canonical operators (which is
similar to that for the LC configurations, provided in Refs. \cite%
{ivv,visin2}), we prove\footnote{%
arranging the right side into groups with three, two and just one
derivatives and consequently it is impossible to have compensations between
them}:

\begin{lemma}
a) For a quantity $\ _{T}\widehat{\mathcal{Q}}:=\widehat{\mathbf{D}}_{\alpha
}\mathbf{K}^{\alpha \beta }\widehat{\mathbf{D}}_{\beta }$, the commutator
{\small
\begin{eqnarray*}
\lbrack \widehat{\square },\ \ _{T}\widehat{\mathcal{Q}}] = 2 (\widehat{\mathbf{D}}^{(\gamma }\mathbf{K}^{\alpha \beta }) \widehat{%
\mathbf{D}}_{\gamma }\widehat{\mathbf{D}}_{\alpha }\widehat{\mathbf{D}}%
_{\beta }+ 3\widehat{\mathbf{D}}_{\mu } ( \widehat{\mathbf{D}}^{(\gamma }%
\mathbf{K}^{\mu \beta )}) \widehat{\mathbf{D}}_{\beta }\widehat{\mathbf{D}}%
_{\gamma }+  \\
\{-\frac{4}{3}\widehat{\mathbf{D}}_{\gamma }( \widehat{\mathbf{R}%
}_{\nu }^{~\ [\gamma }\mathbf{K}^{\beta ]\nu }) +\widehat{\mathbf{D}}_{\gamma } [\frac{1}{2}\mathbf{g}_{\mu \tau }(
\widehat{\mathbf{D}}^{\gamma }\widehat{\mathbf{D}}^{(\mu }\mathbf{K}^{\tau
\beta )}-\widehat{\mathbf{D}}^{\beta }\widehat{\mathbf{D}}^{(\mu }\mathbf{K}%
^{\gamma \tau )}) +\widehat{\mathbf{D}}_{\alpha }\widehat{\mathbf{D}}%
^{(\gamma }\mathbf{K}^{\alpha \beta )}] \}\widehat{\mathbf{D}}_{\beta }.
\end{eqnarray*}%
}
b) In the case of SK d--tensors all the symmetrized derivatives vanish and
\begin{equation}
\lbrack \widehat{\square },\ _{T}\widehat{\mathcal{Q}}]=-\frac{4}{3}\widehat{%
\mathbf{D}}^{\gamma }\left( \widehat{\mathbf{R}}_{\nu }^{~\ [\gamma }\mathbf{%
K}^{\beta ]\nu }\right) \widehat{\mathbf{D}}_{\beta }\,.  \label{comut}
\end{equation}
\end{lemma}

Similar formulas for $\widehat{\mathbf{D}}\rightarrow \nabla $ and SK
tensors positively exhibit quantum anomalies, i. e. the classical
conservation law does not transfer to the quantum level. Even if we evaluate
the commutator for CSK tensors associated with CKY tensors and the LC
connection we do not obtain a cancelation of anomalies \cite{visin2}.
Therefore we are not able to identify any favorable circumstances on the CSK
tensors in order to achieve a conserved quantum operator. This problem seems
to exist for all hidden symmetries of gravitational models with, or not,
torsion.

Nevertheless, there are possibilities to consider such nonholonomic
distributions on a spacetime $\mathbf{V,}$ when the associated nonholonomic
commutator (\ref{comut}) computed for $\widehat{\mathbf{D}}$ vanishes.

\begin{theorem}
\label{th4.1}For any given data $\left( \mathbf{g},\nabla \right) $ and
prescribed N--connection structure $\mathbf{N,}$ we can construct
nonholonomic deformations to some chosen canonical data $\left( \ ^{\eta }%
\mathbf{g},\ ^{\eta }\mathbf{N},\ ^{\eta }\widehat{\mathbf{D}}\right) $ when
$[\ ^{\eta }\widehat{\square },\ _{T}^{\eta }\widehat{\mathcal{Q}}]=0.$ If $%
\mathbf{N}$ is fixed to solve the LC conditions (\ref{lcconstr}) (for a zero
nonholonomically induced canonical torsion), the nonholonomic gravitational
anomaly vanishes for the Einstein spaces, $[\ ^{\eta }\square ,\ \
_{T}^{\eta }\mathcal{Q}]=0.$
\end{theorem}

\begin{proof}
In general, a metric $\mathbf{g}$ is not a solution of Einstein equations
for $\nabla .$ We fix a nonholonomic splitting with a \textquotedblright
primary\textquotedblright\ $\mathbf{N}$ consider nonholonomic deforms of
type (similarly to (\ref{nhdef}) and (\ref{polf}))
\begin{eqnarray}
\ \mathbf{g} &\mathbf{=}&[\ g_{i},\ h_{a},\ N_{i}^{a\,}]\rightarrow
\label{nhdef1} \\
\ ^{\eta }\mathbf{g} &=&[\ ^{\eta }g_{i}=\eta _{i}g_{i},\ ^{\eta }h_{a}=\eta
_{a}h_{a},\ \ ^{\eta }N_{i}^{3\,}=\eta _{i}^{3}w_{i},\ \ ^{\eta
}N_{i}^{4\,}=\eta _{i}^{4}n_{i}],  \notag
\end{eqnarray}%
where $\eta _{\alpha }=1+\chi _{\alpha }=(\eta _{i}=1+\chi _{i}(u^{\beta
}),\eta _{a}=1+\chi _{a}(u^{\beta }),\eta _{i}^{a}=\delta _{i}^{a}+\chi
_{i}^{a}(u^{\beta })$. We do not consider summation on repeating indices in
the above formulas. The gravitational polarizations $\eta _{\alpha }$ and $%
\eta _{i}^{a}$ are uniquely defined if we chose any $\ ^{\eta }\mathbf{g}$
defining an exact solution of
\begin{equation}
\ ^{\eta }\widehat{\mathbf{R}}_{\ \beta \delta }=\lambda \ ^{\eta }\mathbf{g}%
_{\ \beta \delta },  \label{einstcosmsol}
\end{equation}%
see equations (\ref{cdeinst}) for source (\ref{source}) defined by a
nontrivial cosmological constant $\lambda .$ The solutions $\ ^{\eta }%
\mathbf{g}$ can be always constructed following the conditions of Lemma \ref%
{lem} and Theorem \ref{thgs} and (for the LC conditions) Corollary \ref%
{corol1}.

For different classes of solutions $\ ^{\eta }\mathbf{g,}$ we have different
types of hidden nonholonomic symmetries determined by the primary data $%
\left( \mathbf{g},\nabla \right) $ and target data $\left(\ ^{\eta }\mathbf{g%
},\ ^{\eta }\mathbf{N},\ ^{\eta }\widehat{\mathbf{D}}\right)$. If the LC
conditions (\ref{lccond}) are satisfied and $\mathbf{g}$ is also a solution
of the Einstein equations, we can constrain the class of nonholonomic
transforms, i.e. the gravitational polarizations $\eta _{\alpha }$ and $\eta
_{i}^{a}$ in such a form that $\ ^{\eta }\mathbf{g}$ is generated via a
frame transform of $\mathbf{g}$. In both cases, for a chosen class of target
exact solutions, the $\eta $--coefficients characterize the ''flexibility'',
i.e. symmetry of with respect to possible nonholonomic deforms (\ref{nhdef1}%
). For explicit geometric constructions, we can fix the target metrics to be
certain classes of nonholonomic Kerr--Sen black hole configurations in GR
(for instance, with small ellipsoidal and/or solitonic deformations).

If the conditions (\ref{einstcosmsol}) are satisfied for the commutator (\ref%
{comut}), we get $[\ ^{\eta }\widehat{\square },\ _{T}^{\eta }\widehat{%
\mathcal{Q}}]=0,$ even (in general), computing for $\mathbf{g}$ we have $%
[\square ,\ \ _{T}\mathcal{Q}]\neq 0$. We conclude that we can always
eliminate such gravitational anomalies for certain classical symmetries of a
(pseudo) Riemannian metric $\mathbf{g}$ if we suppose that such a metric can
be nonholonomically deformed to a solution of the Einstein equations for a
canonical d--connection $\ ^{\eta }\widehat{\mathbf{D}}$ uniquely defined by
the same metric structure. Such nonholonomic hidden symmetries characterize
both the ''off--diagonal'' nonlinear properties of $\mathbf{g}$ and its
''flexibility'' to be transformed in Einstein spaces. $\square $
\end{proof}

\vskip3pt It is not possible to provide proofs of the above Lemma and
Theorem on nonholonomic cancelation of anomalies (or other results related
to nonholonomic hidden symmetries) \ for spaces endowed with general torsion
structure even there are various attempts and examples related to (super)
string theory \cite{houri,kub,strom,agricola,frolov}. For generic
off--diagonal configurations in GR, such constructions can be performed in a
simplified form if it is known how a result can be derived, for instance,
for a (non) vacuum solution $\ ^{\circ }\mathbf{g}$ and corresponding
LC--connection $\ ^{\circ }\nabla $ using certain vanishing commutators. Via
nonholonomic deformations, $\ ^{\circ }\mathbf{g\rightarrow }\ ^{\eta }%
\mathbf{g}$, we can derive necessary type canonical decomposition $\ ^{\circ
}\nabla \rightarrow \ ^{\eta }\widehat{\mathbf{D}}=\ ^{\eta }\nabla -\
^{\eta }\widehat{\mathbf{Z}}$ generalizing the "anomaly" formulas for $\
^{\eta }\widehat{\mathbf{D}}$ which (in general) contain non-trivial torsion
components.

\begin{remark}
If we select for $\ ^{\eta }\mathbf{g}$ such nonholonomic parametrizations
that the conditions (\ref{einstcosmsol}) and (\ref{lccond}) are satisfied,
we get $[\ ^{\eta }\widehat{\square },\ _{T}^{\eta }\widehat{\mathcal{Q}}]=0$%
, i. e. canceling of certain anomalies for $\ ^{\eta }\widehat{\mathbf{D}} $
$\ $and, for zero torsions, for $\ ^{\eta }\nabla .$ We can always define
such $\ ^{\eta }\widehat{\mathbf{D}}$ \ with nonholonomic hidden symmetries
even we begin with a primary metric $\ ^{\circ }\mathbf{g}$ \ for which $[\
\ ^{\circ }\square ,\ _{T}^{\circ }\mathcal{Q}]\neq 0.$ Computations for
nontrivial $\ ^{\eta }\widehat{\mathbf{T}}$ are much simple and similar to
those for $\ ^{\circ }\nabla ,$ with $\ ^{\circ }T=0,$ because the
distortion tensor $\ ^{\eta }\widehat{\mathbf{Z}}$ from (\ref{distorsrel})
is completely determined by the metric tensor and finally constrained to be
zero.
\end{remark}

We emphasize this important consequence of the above Theorem: We can always
cancel anomalies of a Klein--Gordon operator in GR if we chose a
corresponding \ class of nonholonomic deformations of the frame and linear
connection structures \ uniquely defined \ by the \ metric structure.
Finally, we recover the conditions for the LC configurations \ for some
specials subclasses of nonholonomic deformations.

\section{Conclusions and Discussion}

\label{s6}In this work we have studied a new class of hidden nonholonomic
symmetries of generic off--diagonal solutions in the general relativity
theory. Such Einstein spacetimes can be generated via nonholonomic frame
transforms and deformations of fundamental geometric objects (for instance,
of the linear connection structure and related differential operators,
Riemann and Ricci tensors etc) to certain gravitational configurations with 
 explicit decoupling of some generalized Einstein gravitational
field equations. This allows us to find solutions in very general forms when
the coefficients of off--diagonal metrics depend on some classes of
generating and integration functions. Imposing nonholonomic constraints on
such integral varieties, we select the torsionless Levi--Civita, LC,
configurations and generate solutions in the Einstein gravity theory (such
methods and various examples are provided in Refs. \cite%
{vex2,vex3,vp,vt,vsingl1}).

Via nonholonomic deformations, various types of torsion fields can be
induced on a (pseudo) Riemannian manifold. Such torsions are very different
from those, for instance, in the Einstein--Cartan and/or string gravity
theories because, in our approach, they are completely defined by the
off--diagonal coefficients of the metric tensor following certain well
defined geometric principles. We do not need additional field equations for
such nonholonomically induced torsion fields, as it is considered, for
instance, in \cite{houri}, and all constructions can be performed
equivalently to those for the LC connection.

The gravitational and matter field interactions in the Einstein gravity
theory and modifications have a generic nonlinear character. They are
characterized by various \ types of nonlinear symmetries and different
classes of nonholonomic constraints. An effective tool for study such
classical and quantum interactions is the formalism of Killing--Yano and
Stackel--Killing tensors and their higher order generalizations and
extensions to nontrivial torsion fields, off--diagonal metrics,
non--integrable constraints etc. Their generalized conformal and anholonomic
analogs are related to a multitude of theoretical and mathematical physics
issues such as classical integrability of systems together with their
quantization, supergravity, string theories, hidden symmetries in higher
dimensional black--holes spacetimes, etc.

The considerations in this article are totally classical. Nevertheless,
certain constructions emphasize the existence of important quantum effects because of
possible elimination of certain types of quantum anomalies by classical
nonholonomic transforms/deformations. Let us briefly discuss such issues:\
We found new classes of generic off--diagonal exact solutions in general relativity and modifications 
possessing hidden symmetries generated by corresponding Killing--Yano and
Stackel--Killing tensors. The classical conserved quantities are not
generally preserved when we pass to the quantum mechanical level, i.e.
certain anomalies may occurs. In this work, there are analyzed gravitational
anomalies determined by surviving (in general) of the second term in commutator (\ref%
{commck}) by evaluations  for the solutions of the
Klein--Gordon equation with nonholonomically deformed Hamiltonian (\ref{KG}). It should be emphasized here that we consider the concept of "quantum
gravitational anomaly" in a very restricted sense, i.e. for certain quantum
operators on classical curved spacetimes when we do not have a well defined
and generally accepted model of quantum gravity.  See, for instance, \cite%
{visin2} and, for a review of main concepts on such quantum anomalies related to the wave, Klein--Gordon and Dirac operators, the
section 6 of \cite{vis3}.  The physical importance of Theorem \ref{th4.1} is
that it concludes a procedure of nonholonomic deformations resulting in
vanishing of certain classes of (nonholonomic) gravitational anomalies (even
we are not able to identify any general favorable circumstances of CSK
tensors in GR, and modifications, in order to construct conserved quantum
operators). This allows us to transfer certain classical conserved
quantities on off--diagonal spacetimes to the quantum mechanical level if
necessary types of nonholonomic constraints are imposed on a corresponding
classical nonlinear gravitational dynamics.

Let us speculate on possible quantum effects of "nonholonomic elimination"
of anomalies. \ For the Kerr and Reisner--Nordstr\"{o}m curved--space
backgrounds, there is a relationship between the quantum conservation laws
and the corresponding quantum numbers associated to operators that commute
with the wave operator in a first--quantized field theory. If the Ricci
curvature vanishes for any such type exact solution, this implies the
existence of a Killing vector/tensor in the first-/ second--order cases and
a corresponding constant for theories in the classical limit. Surprisingly,
such a property can be preserved under certain classes of off--diagonal
nonholonomic deformations of solutions (as we proved in our work). Perhaps,
this is explained in some sense as straightforward consequences of something
more fundamental. In \cite{vex4}, we proved that such nonholonomic
deformations can be constructed in a general form depending on finite sets
of parameters inducing corresponding hierarchies of conservation laws which,
in their turn, result in bi--Hamilton and associated solitonic hierarchies
\cite{vsolit}. Such terms contribute efficiently in the Hamiltonian (\ref{KG}%
) determined by nonholonomic deforms of the Klein-Gordon d--operator and may
result in observable quantum effects and new types of conservation laws of
solitonic nature. In section \ref{sssolit}, we provided explicit examples of
solitonic deformations of exact solutions which via nonlinear and
linear connection coefficients distort  the d'Alambert
operator.

Another application is similar to that for the Runge--Lenz constants in the
nonrelativistic hydrogen atom problem. The hidden symmetries of generic
off--diagonal solutions involve certain generalized types of
conserved quantities. If in the "hydrogen atom problem" the Runge--Lenz
constants and conservation law are additional to the well known for the
axial angular momentum, energy, and proper mass, their nonholonomic
deformation generalizations arising with hidden ellipsoid, solitonic,
parametric or other symmetries also may give rise to the degeneracy of
energy levels. We can prescribe a kind of Geroch multi-parametric group for anholonomic deformations  and
associated solitonic parameters \cite{vex4} for certain classes of solutions but, in
general, only the existence of certain hidden symmetries can hardly provide
a satisfactory explanation in itself. It  connects  certain mysteries to
some possible prescribed symmetries. Perhaps, this  provides a generalization
of the "Schiff conjecture" concerning the relationship between classical
mechanics, with nonholonomic constraints, generalizations to nonholonomic
frames in gravity and the equivalence principle, see details in \cite{BC1}.

It is an interesting question how the fact of canceling gravitational
anomalies for certain hidden nonholonomic symmetries demonstrated in this
work can be applied in modern quantum gravity. There are certain
quantization schemes based on introducing alternative connections in the
Einstein gravity theory and generalizations (for instance, the almost K\"{a}%
hler Cartan one, which is also completely defined by the metric structure in
a compatible form), see Refs. \cite{vdq3,vabr,vtwocon}. How such various
bi--connection formulations of a classical and quantum gravitational model
can be related to new insights on nonholonomic canceling of anomalies and
renormalization problem of gravitational interactions is a matter of further
research and elaborating new geometric methods.

Finally, it should be emphasized that an obvious extension of the results of
this paper can be performed for study hidden symmetries for nonholonomic
generalizations of particle like gravitational solutions and
Einstein--Yang--Mills--Higgs systems considered in Refs. \cite{br1,br2,blaga}%
. We should apply the nonlinear connection formalism and the results on
generic off--diagonal Taub NUT configurations, running black holes and
generalized solutions for the Einstein--Dirac solitonic and pp--waves \cite%
{vsingl1,vp,vt}. We plan to address such issues in our further works.

\vskip5pt

\textbf{Acknowledgments:\ } The work is partially supported by the Program
IDEI, PN-II-ID-PCE-2011-3-0256. Author is grateful to N. Mavromatos, P.
Stavrinos and M. Vi\c{s}inescu for important discussions and kind support.

\appendix

\setcounter{equation}{0} \renewcommand{\theequation}
{A.\arabic{equation}} \setcounter{subsection}{0}
\renewcommand{\thesubsection}
{A.\arabic{subsection}}

\section{N--adapted Coefficients for d--Connections}

For convenience, we present some necessary formulas from the geometry of
N--anholonomic (pseudo) Riemannian spaces, see details in \cite{vex3}.

The coefficients of the Levi--Civita (LC) connection $\nabla =\{\
_{\shortmid }\Gamma _{\ \alpha \beta }^{\gamma }\}$ for the metric (\ref{dm}%
) (computed with respect to N--adapted basis (\ref{nader}) and (\ref{nadif}%
)) can be written in the form
\begin{equation}
\ _{\shortmid }\Gamma _{\ \alpha \beta }^{\gamma }=\widehat{\mathbf{\Gamma }}%
_{\ \alpha \beta }^{\gamma }+Z_{\ \alpha \beta }^{\gamma }.  \label{deflc}
\end{equation}%
The value $\widehat{\mathbf{D}}=\{\widehat{\mathbf{\Gamma }}_{\ \alpha \beta
}^{\gamma }=(\widehat{L}_{jk}^{i},\widehat{L}_{bk}^{a},\widehat{C}_{jc}^{i},%
\widehat{C}_{bc}^{a})\}$ with coefficients {\small
\begin{eqnarray}
\widehat{C}_{jc}^{i} &=&\frac{1}{2}g^{ik}e_{c}g_{jk},\ \widehat{L}_{bk}^{a}
= e_{b}(N_{k}^{a})+\frac{1}{2}h^{ac} (e_{k}h_{bc}-h_{dc}\
e_{b}N_{k}^{d}-h_{db}\ e_{c}N_{k}^{d}) ,  \notag \\
\widehat{L}_{jk}^{i} &=&\frac{1}{2}g^{ir}
(e_{k}g_{jr}+e_{j}g_{kr}-e_{r}g_{jk}), \widehat{C}_{bc}^{a}= \frac{1}{2}%
h^{ad} (e_{c}h_{bd}+e_{c}h_{cd}-e_{d}h_{bc})  \label{candcon}
\end{eqnarray}%
} defines the canonical distinguished connection (d--connection). By
straightforward computations, we can check that it is metric compatible, $%
\widehat{\mathbf{D}}\mathbf{g}=0,$ and its torsion $\mathcal{T}=\{\widehat{%
\mathbf{T}}_{\ \alpha \beta }^{\gamma }\equiv \widehat{\mathbf{\Gamma }}_{\
\alpha \beta }^{\gamma }-\widehat{\mathbf{\Gamma }}_{\ \beta \alpha
}^{\gamma };\widehat{T}_{\ jk}^{i},\widehat{T}_{\ ja}^{i},\widehat{T}_{\
ji}^{a},\widehat{T}_{\ bi}^{a},\widehat{T}_{\ bc}^{a}\},$ is with zero
horizontal and vertical coefficients, $\widehat{T}_{\ jk}^{i}=0$ and $%
\widehat{T}_{\ bc}^{a}=0.$ There are also nontrivial h--v-- coefficients
\begin{eqnarray}
\widehat{T}_{\ jk}^{i} &=&\widehat{L}_{jk}^{i}-\widehat{L}_{kj}^{i},\widehat{%
T}_{\ ja}^{i}=\widehat{C}_{jb}^{i},\widehat{T}_{\ ji}^{a}=-\Omega _{\
ji}^{a},  \label{dtors} \\
\widehat{T}_{aj}^{c} &=&\widehat{L}_{aj}^{c}-e_{a}(N_{j}^{c}),\widehat{T}_{\
bc}^{a}=\ \widehat{C}_{bc}^{a}-\ \widehat{C}_{cb}^{a}.  \notag
\end{eqnarray}

The distortion tensor $Z_{\ \alpha \beta }^{\gamma }$ in (\ref{deflc}) is
also constructed in a unique form from the coefficients of metric
N--connection, {\small
\begin{eqnarray}
\ Z_{jk}^{i} &=&0,\ Z_{jk}^{a}=-\widehat{C}_{jb}^{i}g_{ik}h^{ab}-\frac{1}{2}%
\Omega _{jk}^{a},~Z_{bk}^{i}=\frac{1}{2}\Omega _{jk}^{c}h_{cb}g^{ji}-\Xi
_{jk}^{ih}~\widehat{C}_{hb}^{j},  \notag \\
Z_{bk}^{a} &=&\ ^{+}\Xi _{cd}^{ab}~\widehat{T}_{kb}^{c},\ Z_{kb}^{i}=\frac{1%
}{2}\Omega _{jk}^{a}h_{cb}g^{ji}+\Xi _{jk}^{ih}~\widehat{C}_{hb}^{j},\
Z_{bc}^{a}=0,\   \notag \\
Z_{jb}^{a} &=&-\ ^{-}\Xi _{cb}^{ad}~\widehat{T}_{jd}^{c},\ Z_{ab}^{i}=-\frac{%
g^{ij}}{2}\left[ \widehat{T}_{ja}^{c}h_{cb}+\widehat{T}_{jb}^{c}h_{ca}\right]
,  \label{deft}
\end{eqnarray}%
for $\ \Xi _{jk}^{ih}=\frac{1}{2}(\delta _{j}^{i}\delta
_{k}^{h}-g_{jk}g^{ih})$ and $~^{\pm }\Xi _{cd}^{ab}=\frac{1}{2}(\delta
_{c}^{a}\delta _{d}^{b}+h_{cd}h^{ab}).$ }

Any geometric and physical formulas for the connection $\nabla $ can be
equivalently redefined for the canonical d--connection $\widehat{\mathbf{D}}%
, $ and inversely, using (\ref{deflc}) because all involved geometric
objects (two different connections and the distorting tensor) are uniquely
defined by the same metric structure.

\setcounter{equation}{0} \renewcommand{\theequation}
{B.\arabic{equation}} \setcounter{subsection}{0}
\renewcommand{\thesubsection}
{B.\arabic{subsection}}

\section{Decoupling and Integration of Einstein Eqs}

\label{asec2}We briefly summarize the results on generating off--diagonal
solutions in gravity \cite{vex2,vex3}.

Using the conditions of Lemma \ref{lem} and computing in explicit form the
Ricci and Einstein tensors, we prove:

\begin{theorem}
(\textbf{decoupling of equations}). The Einstein equations for $\widehat{%
\mathbf{D}}$ (\ref{candcon}) and ansatz for the metric $\mathbf{g}$ (\ref%
{ans1}) with $\omega =1$ and any general source $\mathbf{\Upsilon }_{\
\delta }^{\alpha }$ (\ref{source}) are {\small
\begin{eqnarray}
&&\widehat{R}_{1}^{1}=\widehat{R}_{2}^{2}=-\frac{1}{2g_{1}g_{2}}%
[g_{2}^{\bullet \bullet }-\frac{g_{1}^{\bullet }g_{2}^{\bullet }}{2g_{1}}-%
\frac{\left( g_{2}^{\bullet }\right) ^{2}}{2g_{2}}+g_{1}^{\prime \prime }-%
\frac{g_{1}^{\prime }g_{2}^{\prime }}{2g_{2}}-\frac{\left( g_{1}^{\prime
}\right) ^{2}}{2g_{1}}]=-\Upsilon _{2}(x^{k}),  \notag \\
&&\widehat{R}_{3}^{3}=\widehat{R}_{4}^{4}=-\frac{1}{2h_{3}h_{4}}[h_{4}^{\ast
\ast }-\frac{\left( h_{4}^{\ast }\right) ^{2}}{2h_{4}}-\frac{h_{3}^{\ast
}h_{4}^{\ast }}{2h_{3}}]=-\Upsilon _{4}(x^{k},y^{3}),  \label{eqe} \\
&&\widehat{R}_{3k}=\frac{w_{k}}{2h_{4}}[h_{4}^{\ast \ast }-\frac{%
(h_{4}^{\ast })^{2}}{2h_{4}}-\frac{h_{3}^{\ast }h_{4}^{\ast }}{2h_{3}}]+%
\frac{h_{4}^{\ast }}{4h_{4}}(\frac{\partial _{k}h_{3}}{h_{3}}+\frac{\partial
_{k}h_{4}}{h_{4}})-\frac{\partial _{k}h_{4}^{\ast }}{2h_{4}}=0  \notag \\
&&\widehat{R}_{4k}=\frac{h_{4}}{2h_{3}}n_{k}^{\ast \ast }+\left( \frac{h_{4}%
}{h_{3}}h_{3}^{\ast }-\frac{3}{2}h_{4}^{\ast }\right) \frac{n_{k}^{\ast }}{%
2h_{3}}=0.  \notag
\end{eqnarray}%
}
\end{theorem}

The system of partial differential equations (\ref{eqe}) is with decoupling
of equations (the "splitting" of equations is used as an equivalent one; we
should not confuse this with the property of separation of variables). For
simplicity, we can consider a subclass of solutions when for the chosen
N--system of reference $h_{a}^{\ast }\neq 0.$\footnote{%
If $h_{3}^{\ast }=0,$ or $h_{4}^{\ast }=0,$ the solutions can be constructed
similarly (in certain cases, they can be transformed from one to another one
via frame/coordinate transforms).}

\begin{corollary}
The system of equations (\ref{eqe}) for $y^{3}=v,$ $g_{i}=\epsilon
_{i}e^{\psi (x^{k})}$ and $h_{a}^{\ast }\neq 0,$ $\Psi _{2}\neq 0,$ $\Psi
_{4}\neq 0,$ can be written equivalently in the form%
\begin{eqnarray}
\psi ^{\bullet \bullet }+\psi ^{\prime \prime } &=&2\Psi _{4}(x^{k})
\label{eqe1} \\
h_{4}^{\ast } &=&2h_{3}h_{4}\Psi _{2}(x^{k},v)/\phi ^{\ast }, \ \beta
w_{i}+\alpha _{i} = 0,\ n_{i}^{\ast \ast }+\gamma n_{i}^{\ast } = 0,  \notag
\end{eqnarray}%
\begin{equation*}
\mbox{ for }\phi =\ln | \frac{h_{4}^{\ast }}{\sqrt{|h_{3}h_{4}|}}| ,\ \gamma
:=(\ln \frac{|h_{4}|^{3/2}}{|h_{3}|}) ^{\ast }, \ \alpha _{i}=h_{4}^{\ast
}\partial _{i}\phi ,\ \beta =h_{4}^{\ast }\phi ^{\ast }.
\end{equation*}
\end{corollary}

The systems of equations (\ref{eqe}) and (\ref{eqe1}) can be integrated in
general forms following the results of Theorem \ref{thgs}; we can generate
solutions for the LC connection if the zero torsion conditions of Corollary %
\ref{corol1} are satisfied.

Let us study the ''vanishing torsion'' conditions (\ref{lccond}). For
general sources, $\mathbf{\Upsilon }_{\ \delta }^{\alpha },$ it is quite
difficult to prove in an explicit analytic form that such equations have
nontrivial solutions.

\begin{corollary}
\label{corolap1}We can adapt the nonholonomic distributions for generic
off--diagonal Einstein spaces with $\mathbf{\Upsilon }_{\ \delta }^{\alpha
}=\lambda \mathbf{\delta }_{\ \delta }^{\alpha },$ when $\Psi _{2}=\Psi
_{4}=\lambda ,$ and parametrize the data (\ref{data1}) and (\ref{data2}) for
the coefficients of metric ansatz in such a form that (\ref{lccond})
determine some classes of nontrivial solutions, when the d--torsions (\ref%
{dtors}) for $\widehat{\mathbf{D}}$ are zero.
\end{corollary}

\begin{proof}
Let us consider a solution (\ref{data1}) and (\ref{data2}) when the
coordinate system and boundary conditions are fixed in the form that $\
\underline{h}_{4}(x^{k})=0,$ $\ ^{2}n_{k}(x^{i})=0$ and $\partial _{i}\
^{1}n_{j}(x^{k})=\partial _{j}\ ^{2}n_{i}(x^{k}).$ In such cases, we must
prove that \ $h_{4}=\ \pm \lambda ^{-1}e^{\phi }$ and $w_{i}= \partial
_{i}\phi /\phi ^{\ast }$ have for some classes of functions $\phi (x^{i},v),$
$\phi ^{\ast }\neq 0,$ certain nontrivial solutions of (\ref{lccond}), i.e.
\begin{equation}
w_{i}^{\ast }=\mathbf{e}_{i}\ln |\ h_{4}|\mbox{ and }\partial
_{i}w_{j}=\partial _{j}w_{i}.  \label{aux}
\end{equation}%
Expressing $h_{4}$ and $w_{i}$ explicitly via $\phi ,$ we get from (\ref{aux}%
) that $\phi ^{\ast \ast }\partial _{i}\phi -\phi ^{\ast }(\partial _{i}\phi
)^{\ast }=0$. As a particular case, these equations can be written in the
form $ \left( \partial _{i}\phi /\phi ^{\ast }\right) ^{\ast }=w_{i}^{\ast
}=0,$ when $w_{i}=\ ^{0}w_{i}(x^{k}).$ Choosing $\phi =\ ^{0}\phi (x^{k})%
\overline{\phi (}v),$ we generate solutions with separation of variables, $\
^{0}w_{i}(x^{k})=\iota \partial _{i}\ ^{0}\phi $ and $\overline{\phi }^{\ast
}=\iota ^{-1}\overline{\phi },$ for a nonzero constant $\iota .$ In general,
\ we can consider various types of nonholonomic constraints (\ref{lccond})
for selecting from data (\ref{data1}) and (\ref{data2}) very different
families of solutions in GR. $\square $
\end{proof}

\begin{remark}
\label{remarka}Integrating on variable $y^{3}=\varphi $ and taking $%
\underline{h}_{4}=\ ^{0}h_{4}(\widetilde{x}^{i}),$ $\epsilon _{i}=1,$ $%
g_{i}=\eta _{i}e^{\psi (\widetilde{x}^{k})}\ ^{0}g_{i}(\widetilde{x}%
^{k}),\eta _{i}=\epsilon _{i}e^{\psi (\widetilde{x}^{k})},\ \omega =1,$ in
coordinates $\widetilde{u}^{\beta }$ (\ref{coordi}) and for the ''initial''
data $\ ^{\circ }\mathbf{g}_{\alpha }$ (\ref{datai}), the ''one--Killing''
off--diagonal solutions (\ref{data1}) and (\ref{data2}) for the Einstein
spaces can be represented in the form
\begin{eqnarray}
&&\epsilon _{i}\frac{\partial ^{2}}{(\partial \widetilde{x}^{i})^{2}}[\psi (%
\widehat{\Phi }+2\ln |\ ^{b}\rho |)]=\lambda ,  \label{data1b} \\
h_{4} &=&\ \ ^{\circ }h_{4}\pm \lambda ^{-1}e^{2\phi },\ h_{3}=\left[ \left(
\sqrt{|h_{4}|}\right) ^{\ast }\right] ^{2}e^{-2\phi },  \notag \\
w_{i} &=&-(\phi ^{\ast })^{-1}\frac{\partial \phi }{\partial \widetilde{x}%
^{i}},\ n_{k}=\ ^{1}n_{k}+\ ^{2}n_{k}\int d\varphi \ h_{3}/\left( \sqrt{%
|h_{4}|}\right) ^{3},  \notag
\end{eqnarray}%
for respectively given $\ ^{\circ }h_{4}(\widetilde{x}^{i}),\widehat{\Phi }(%
\widetilde{x}^{i}),\ ^{b}\rho (\widetilde{x}^{i})$ and chosen generating/
integration functions $\phi (\widetilde{x}^{i},\varphi ),$ $^{1}n_{k}(%
\widetilde{x}^{i}),\ ^{2}n_{k}(\widetilde{x}^{i})$ subjected to the
conditions (\ref{aux}).
\end{remark}

To consider possible small nonholonomic deformations $\ ^{\circ }\mathbf{g}%
_{\alpha }\rightarrow $ $\ ^{\eta }\mathbf{g}_{\alpha }$ is convenient to
express the formulas (\ref{aux}) via polarization functions (\ref{polf})
when, for $\lambda \neq 0,$ $h_{a}=\ ^{\circ }h_{a}\eta _{a}=\ ^{\circ
}h_{a}(1+\chi _{a})$. Using the second formula in (\ref{data1b}), for $%
h_{4}, $ we can express $\phi =\ln \sqrt{|\lambda \ ^{\circ }h_{4}\chi _{4}|}
$ and consider for such classes of solutions a ''new'' generating function $%
\chi _{4}(\widetilde{x}^{i},\varphi ).$ The third formula for $h_{3}$ from
that system of solutions allows us to express
\begin{equation}
\chi _{3}=-1+(\lambda \ ^{\circ }h_{3})^{-1}\left[ \left( \ln \sqrt{|1-\chi
_{4}|}\right) ^{\ast }\right] ^{2}.  \label{aux4a}
\end{equation}%
We conclude that all $v$--coefficients of off--diagonal metrics and
N--connections, see ansatz (\ref{ans1}) (equivalently (\ref{ans1a})) for the
Einstein spaces with Killing symmetry on $\partial /\partial \widetilde{y}%
^{4},$ up to arbitrary frame transforms, are functionally determined by $%
\chi _{4},$ i.e. $\phi \lbrack \chi _{4}],h_{a}$ $[\chi _{4}],$\ $w_{i}[\chi
_{4}],n_{k}[\chi _{4}].$

\end{document}